\documentclass[onefignum,onetabnum]{siamart190516}

\usepackage{lipsum}
\usepackage{amsfonts,amssymb}
\usepackage{graphicx}
\usepackage{epstopdf}
\usepackage{algorithmic}
\usepackage{color}
\usepackage{ulem}
\usepackage{comment}

\newcommand{\R}{\mathbb{R}}

\newcommand{\N}{\mathbb{N}}

\newcommand{\A}{\mathcal{A}}
\newcommand{\E}{\mathcal{E}}
\newcommand{\G}{\mathcal{G}}
\newcommand{\Q}{\mathcal{Q}}

\usepackage{framed}

\ifpdf
  \DeclareGraphicsExtensions{.eps,.pdf,.png,.jpg}
\else
  \DeclareGraphicsExtensions{.eps}
\fi

\newsiamremark{remark}{Remark}
\newsiamremark{hypothesis}{Hypothesis}
\crefname{hypothesis}{Hypothesis}{Hypotheses}
\newsiamthm{claim}{Claim}
\newsiamthm{example}{Example}
\headers{Kron Reduction and Effective Resistance of Directed Graphs}{T. SUGIYAMA AND K. SATO}

\title{Kron Reduction and Effective Resistance of Directed Graphs\thanks{Submitted to the editors DATE.
\funding{This work was supported by Japan Society for the Promotion of Science KAKENHI under 20K14760.}}}

\author{Tomohiro Sugiyama\thanks{Department of Mathematical Informatics, Graduate School of Information Science and Technology, the University of Tokyo,
  (T. Sugiyama: \email{working-sugiyama@g.ecc.u-tokyo.ac.jp}, K. Sato: \email{kazuhiro@mist.i.u-tokyo.ac.jp})}
\and Kazuhiro Sato\footnotemark[2]}

\usepackage{amsopn}

\ifpdf
\hypersetup{
  pdftitle={An Example Article},
  pdfauthor={D. Doe, P. T. Frank, and J. E. Smith}
}
\fi

\begin{document}

\maketitle
\begin{abstract}
In network theory, the concept of effective resistance is a distance measure on a graph that relates the global network properties to individual connections between nodes.
In addition, the Kron reduction method is a standard tool for reducing or eliminating the desired nodes, which preserves the interconnection structure and the effective resistance of the original graph.
Although these two graph-theoretic concepts stem from the electric network on an undirected graph, they also have a number of applications throughout a wide variety of other fields.
In this study, we propose a generalization of a Kron reduction for directed graphs.
Furthermore, we prove that this reduction method preserves the structure of the original graphs, such as the strong connectivity or weight balance.
In addition, we generalize the effective resistance to a directed graph using Markov chain theory, which is invariant under a Kron reduction.
Although the effective resistance of our proposal is asymmetric, we prove that it induces two novel graph metrics in general strongly connected directed graphs.
In particular, the effective resistance captures the commute and covering times for strongly connected weight balanced directed graphs.
Finally, we compare our method with existing approaches and relate the hitting probability metrics and effective resistance in a stochastic case. In addition, we show that the effective resistance in a doubly stochastic case is the same as the resistance distance in an ergodic Markov chain.
\end{abstract}

\begin{keywords}
  Kron Reduction, Effective Resistance, Directed Graph, Algebraic Graph Theory
\end{keywords}

\begin{AMS}
  	05C50, 05C20, 15A06
\end{AMS}

\section{Introduction}
\label{sec: Introduction}
\subsection{Background}
\label{subsec: Background}
Large-scale network systems are ubiquitous and play a crucial role in modern society, among which
electrical networks, smart power grids \cite{pasqualetti2014controllability}, social networks \cite{liu2011controllability}, and multi-agent systems \cite{mesbahi2010graph} are but a few examples.
However, it is difficult to analyze models, run simulations, or design an appropriate controller for large-scale systems.
The network structure of such systems is often complex, which leads to serious scalability issues owing to the limited computational and storage capacity.
Thus, it is extremely important to have a methodology for graph reduction.
Clustering is one technique for constructing a reduced graph. Here, the vertices of the graph are merged into clusters by considering the edges of the graph \cite{schaeffer2007graph}.
However, if we focus only on important nodes and their connections, node elimination techniques is an appropriate approach.
In this study, we focus on the classical node elimination method: Kron reduction.

Originating from circuit theory, a Kron reduction is a method for simplifying electrical networks while preserving their electrical behavior \cite{kron1939tensor}.
Essentially, the Kron reduction of a connected graph is a Schur complement of the corresponding loopy Laplacian matrix (defined later) with respect to a subset of nodes.
Such a reduction appears in the context of electrical impedance tomography and in power networks \cite{pai2012energy}, among other areas.
As a notable property of the Kron reduction method, it preserves the effective resistance \cite{dorfler2012kron}.
The effective resistance between any two vertices $a$ and $b$ is defined as the resistance of the entire system when a voltage source is connected across them in an electrical network constructed from a graph by replacing each edge with a resistor.
As one of the useful properties of an effective resistance, it defines a distance function that considers the impact on all parallel paths, which is usually different from the shortest path \cite{klein1993resistance}.
This allows the effective resistance to be used in place of the shortest path distance when analyzing problems not only in circuit theory but also in chemistry \cite{janezic2015graph}, control theory \cite{barooah2006graph}, and other areas.
In addition, it is well known that the effective resistance is related to the study of random walks and Markov chains in a network \cite{chandra1996electrical, doyle1984random}.

Based on these definitions, a Kron reduction and an effective resistance are restricted to undirected graphs.
However, in many applications including Markov chains and network systems, directed graphs naturally arise.
Indeed, in \cite{young2015new}, the authors proposed a generalized definition of an effective resistance, and in \cite{fitch2019effective}, the author applied this definition to a Kron reduction of a directed graph while preserving such resistance.
However, these methods are defined only for loop-less directed graphs, and the reduction method may lose the positivity of the edges and interconnection structure of the original graph, as pointed out in Section \ref{sec: Comparison with existing studies}.
Furthermore, these methods have no clear physical interpretations.
\subsection{Contribution}
\label{subsec: Contribution}
First, we generalize the Kron reduction method analyzed in \cite{dorfler2012kron} to include directed graphs and relate the topological and algebraic properties of the resulting Kron-reduced Laplacian to those of the original Laplacian. 
The results provide a foundation for the application of a Kron reduction to network-reduction problems involving directed graphs.
Moreover, we define the effective resistance of directed graphs based on the Markov chain theory using such a reduction.
We show that the effective resistance of a loop-less-directed graph is invariant under a Kron reduction.
Moreover, this can be computed from the pseudoinverse of the loop-less Laplacian and is related to the commute and cover times for strongly connected weight balanced directed graphs.
Although the effective resistances are generally asymmetric, we derive two novel distances based on the effective resistance and node characteristics in a strongly connected case.
Moreover, we provide a comparison with several related existing studies and prove that the effective resistance in a doubly stochastic case is the same as the resistance distance in the ergodic Markov chain defined in \cite{choi2019resistance}.
In addition, we show that the hitting probability metric defined in \cite{boyd2021metric} and the effective resistance in the strongly connected directed case are related.

Table \ref{tab: result} summarizes the properties of the Kron reduction method generalized herein and in \cite{dorfler2012kron, fitch2019effective}.
The proposed reduction method is defined for directed graphs that may contain self-loops.
Furthermore, as previously mentioned, our proposed method preserves the interconnection structure and positivity of the edges of the original graph.
Note that the effective resistance defined in this study is different from the existing definition in \cite{young2015new}, although both definitions coincide in the case of an undirected graph.
Thus, the effective resistance in this study is another generalization of the effective resistance defined in \cite{dorfler2012kron}.
\begin{table}[htbp]
\caption{Properties of our reduction method and those of previous approaches.}
    \label{tab: result}
    \begin{tabular}{|l||l|l|l|}
    \hline
    method & \cite{dorfler2012kron} & \cite{fitch2019effective} & proposed method \\ \hline\hline
    undirected graph& true & true & true \\ \hline
    directed graph & false & true & \textbf{true} \\ \hline
    interconnection structure & true & false & \textbf{true} \\ \hline
    self-loop & true & false & \textbf{true} \\ \hline
    positivity & true & false & \textbf{true} \\ \hline
    \end{tabular}
\end{table}
\subsection{Outline}
The remainder of this paper is organized as follows. In Section \ref{sec: Notation}, we provide some basic concepts of algebraic graph theory.
In Section \ref{sec: Kron reduction of directed graphs}, we introduce a Kron reduction to a directed graph and discuss the aspects of a graph-theoretic analysis.
In Section \ref{sec: Effective resistance of directed graphs}, we provide a novel definition of the effective resistance of directed graphs and detail their properties. 
In Section \ref{sec: Comparison with existing studies}, we compare the above two notions with those of existing studies.
In Section \ref{sec: Conclusion}, we summarize our research.

\section{Notation}
\label{sec: Notation}
Let $\mathbf{1}_{p \times q},\mathbf{0}_{p \times q}$ be $p \times q$ dimensional matrices of a unit and zero entries, let $e_i$ be the vector of
appropriate dimension with $1$ at position $i$ and $0$ at other positions, and let $I_n$ be the $n$-dimensional identity matrix. 
We adopt $\mathbf{1}_n = \mathbf{1}_{n \times 1}$ and $\mathbf{0}_n =\mathbf{0}_{n \times 1}$ as shorthand.
Given a finite set $X$, let $|X|$ be its cardinality, and define for $n \in \N$ the set $[n] = \{1,\ldots, n\}$. 
For $\alpha, \beta \subset [n]$ and $A \in \R^{n \times n}$, let $A[\alpha, \beta]$ denote the submatrix of $A$ obtained by the rows indexed using $\alpha$ and the columns indexed using $\beta$. 
Similarly, for $\alpha \subset [n]$ and $v \in \R^n$, let $v[\alpha] \in \R^{|\alpha|}$ denote the sub-vector of $v$ obtained from the elements indexed using $\alpha$.
The \textit{Schur complement} of $A$ with respect to the indices $\alpha$ is the $|\alpha| \times |\alpha|$ dimensional matrix $A/A[\alpha^c, \alpha^c]$ defined as
\begin{align*}
    A/A[\alpha^c, \alpha^c] &:= A[\alpha, \alpha] - A[\alpha, \alpha^c] A[\alpha^c, \alpha^c]^{-1} A[\alpha^c, \alpha]
\end{align*}
if $A[\alpha^c,\alpha^c]$ is nonsingular, where $\alpha^c = [n] \setminus \alpha$. As an exception, we define $A/A[\emptyset, \emptyset] = A$.
A square matrix $A \in \R^{n \times n}$ is diagonally dominant if $|A_{i i}| \ge \sum_{j \neq i}|A_{i j}|$ for all $i \in [n]$. 
A Z-matrix is a matrix with non-positive off-diagonal entries, whereas an M-matrix is a Z-matrix with eigenvalues, the real parts of which are positive.

Let $\G = ([n], \E, \A)$ be a directed and weighted graph, where $[n]$, $\E \subset [n] \times [n]$, and $\A \in \R^{n \times n}$ denote the node set, edge set, and weighted adjacency matrix with non-negative entries $\A_{i j}$, respectively.
A positive off-diagonal element $\A_{i j} > 0$ induces a weighted edge $\{i,j\} \in \E$ and a positive diagonal element $\A_{ii} > 0$ induces a weighted self-loop $\{i,i\} \in \E$.
The degree matrix $\mathcal{D} \in \R^{n \times n}$ for $\G$ is defined as $\mathcal{D} := \mathrm{diag}(\{\sum_{j=1}^n \A_{i j}\}_{i=1}^n)$, where $\mathrm{diag}(v_1, \ldots, v_n)$ denotes the associated diagonal matrix, that is, $\mathrm{diag}(v_1, \ldots, v_n) =\begin{bmatrix}v_1\\& \ddots\\& &v_n\end{bmatrix} $.
The associated directed (loop-less) Laplacian matrix is defined as $\mathcal{L} := \mathcal{D} - \A$.
Note that self-loops $\A_{ii}$ do not appear in the Laplacian.
For these reasons, we define the loopy Laplacian matrix as $\mathcal{Q} := \mathcal{L} + \mathrm{diag}(\{\A_{ii}\}_{i=1}^n)$ , which is also called a discrete Schr\"odinger operator \cite{bendito2005potential}.
Note that adjacency matrix $\A$ can be recovered from the loopy Laplacian $\Q$ as $\A_{i,i} = \sum_{j=1}^n \Q_{i j}$ and $\A_{i j} = -\Q_{i j}\,(i \neq j)$, and thus $\Q$ uniquely determines the graph $\G$.
We refer to $\Q$ as strictly loopy (loop-less) Laplacian if the associated graph features at least one (no) self-loop.
Based on the Gershgorin circle theorem \cite{horn2012matrix}, all eigenvalues of a loopy Laplacian matrix have non-negative real parts, and the strictly loopy Laplacian matrix is nonsingular.
Let $\Pi_n = I_n - \mathbf{1}_{n \times n}/ n$ denote the orthogonal projection matrix onto the subspace of $\mathbb{R}^n$ perpendicular to $\mathbf{1}_n$.
The matrix $\Pi_n$ is symmetric and since $\mathcal{L}\mathbf{1}_n = \mathbf{0}_n$, $\mathcal{L}\Pi_n = \mathcal{L}$ and $\Pi_n \mathcal{L}^{\top} = \mathcal{L}^{\top}$ for any graph.
A directed graph is strongly connected if there is a path between all pairs of vertices. A square matrix $X \in \R^{n\times n}$ is said to have the strongly connected (SC) property if for each pair
of distinct integers $i,j \in [n]$ there is a sequence of distinct integers $k_1 = i,\ldots, k_m = j$ such that each entry $X_{k_1 k_2} , X_{k_2 k_3} ,\ldots, X_{k_{m-1} k_m} $ is nonzero.
Note that a directed graph $\G$ is strongly connected if and only if any one of $\A$, $\Q$, and $\mathcal{L}$ has the SC property. A matrix  having the SC property is called to be irreducible.

\section{Kron reduction of directed graphs}
\label{sec: Kron reduction of directed graphs}
In \cite{dorfler2012kron}, the authors provided a graph-theoretic analysis of the Kron reduction process of an undirected graph.
In this section, we generalize the definition of a Kron reduction to directed graphs and prove its topological and algebraic properties.
\subsection{Definition}
\label{subsec: Definition}
In this subsection, we define the Kron reduction of a directed graph and provide its structural properties.
\begin{definition}[Kron Reduction]\label{def: KR}
    Let $\Q \in \R^{n \times n}$ be a loopy Laplacian matrix and $\alpha \subsetneq [n]$ be a proper subset of nodes with $|\alpha| \geq 2$. The $|\alpha| \times |\alpha|$ dimensional Kron-reduced matrix $\Q_{\mathrm{red}}$ is defined as
    \begin{align*}
        \Q_{\mathrm{red}} &:= \Q/\Q[\alpha^c, \alpha^c]=\Q[\alpha, \alpha] - \Q[\alpha, \alpha^c] \Q[\alpha^c, \alpha^c]^{-1} \Q[\alpha^c, \alpha].
    \end{align*}
\end{definition}
Note that the definition coincides with that of \cite{dorfler2012kron} for undirected graphs.

In the following, we refer to the nodes of $\alpha$ and $\alpha^c$ as boundary and interior nodes, respectively.
 For notational simplicity and without loss of generality, we will assume that $n$ nodes are labeled such that $\alpha = \{1, \ldots, |\alpha|\}$ and $\alpha^c = \{|\alpha|+1, \ldots, n\}$.
 
In \cite{dorfler2012kron}, it was shown that the Kron-reduced matrix $\Q_{\mathrm{red}}$ is well-defined when $\G$ is undirected and connected.
When $\G$ is directed, however, it is not always well-defined because of the singularity of $\Q[\alpha^c,\alpha^c]$.
To describe the conditions for the existence of a Kron reduction, we introduce the notion of a \textit{reachable subset}.
\begin{definition}[Reachable subset]
Let $\G = ([n], \E, \A)$ be a directed and weighted graph.
The subset $\alpha \subset [n]$ is said to be reachable if for any $i \in \alpha^c$, there exist a node $j \in \alpha$ and a path in $\G$ from $i$ to $j$.
We call node $k$ of $\G$ globally reachable if $\{k\}$ is reachable.
\end{definition}

For example, consider a 3-node graph shown in Fig. \ref{fig: not_transitive}. Although $\{1,2\}$ and $\{2,3\}$ are reachable subsets, $\{1,3\}$ is not, because node $2$ has no outer-edges.
\begin{figure}[hbtp]
    \centering
    \includegraphics[keepaspectratio, scale=0.50]{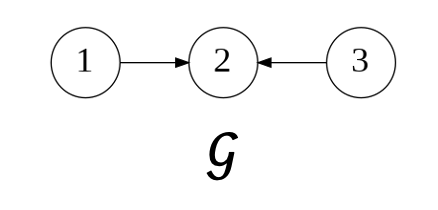}
    \caption{A simple 3-node directed graph.}
    \label{fig: not_transitive}
\end{figure}

The following Lemma is a generalization of Lemma III. 3 in \cite{dorfler2012kron}.
\begin{lemma}[\textit{Well-definedness of Kron Reduction}]\label{lem: Well-definedness of Kron Reduction}
Let $\G = ([n], \E, \A)$ be a directed and weighted graph and $\alpha \subsetneq [n]$ be a proper subset of nodes with $|\alpha|\ge 2$.
Then, the Kron-reduced matrix $\Q_{\mathrm{red}} = \Q/\Q[\alpha^c, \alpha^c]$ is well-defined if $\alpha \subsetneq [n]$ is a reachable subset of $\G$. 
Furthermore, $\alpha$ is a reachable subset if and only if the matrix $\mathcal{L}/\mathcal{L}[\alpha^c,\alpha^c]$ is well-defined.
\end{lemma}
\begin{proof}
    First, consider the case when the subgraph among the interior nodes is strongly connected.
    From the assumption that $\alpha$ is reachable, there exist an interior node $i \in \alpha^c$, a boundary node $j \in \alpha$, and an edge from $i$ to $j$.
     Thus,
     \begin{align*}
         \Q_{i i} = \sum_{k=1}^n \A_{i k}
         \geq \sum_{k \in \alpha^c \setminus \{i\}}\A_{i k} + \A_{i i} + \A_{i j}
          > \sum_{k \in \alpha^c \setminus \{i\}}|\Q_{i k}|,
     \end{align*}
     where $\A_{i j} > 0$ and $\A_{i k} = -\Q_{i k} = |-\Q_{i k}|$ $(k \neq i)$.
     Since $\Q$ is a loopy Laplacian and $\Q[\alpha^c,\alpha^c]$ has the SC property, $\Q[\alpha^c, \alpha^c]$ is irreducibly diagonally dominant, that is, $\Q[\alpha^c, \alpha^c]$ is M-matrix and invertible \cite{horn2012matrix}. 
     If the graph among the interior nodes consists of multiple strongly connected components, then, after relabeling the nodes, the matrix $\Q[\alpha^c, \alpha^c]$ is block trigonal with irreducible diagonal blocks corresponding to the strongly connected components.\footnote{If $\G$ is undirected, $\Q[\alpha^c,\alpha^c]$ can be transformed into a block diagonal form.} Each diagonal block is invertible by the previous arguments, so is $\Q[\alpha^c, \alpha^c]$. Thus, $\Q_{\mathrm{red}}$ is well-defined.\\
    \indent Similar to the previous arguments, $\mathcal{L}/\mathcal{L}[\alpha^c, \alpha^c]$ is well-defined if $\alpha$ is a reachable subset. To show the converse, we assume that $\alpha$ is not a reachable subset. Then, there exists $i \in \alpha^c$ that cannot reach any $j \in \alpha$. Let $\beta$ be the node set of every node in $\G$ that is reachable from node $i$. Since $\beta \subset \alpha^c$ based on such assumption, after relabeling the nodes, $\mathcal{L}$ can be rewritten as
    \begin{align}
    \label{unreachable}
        \mathcal{L} &=
        \left[\begin{array}{c|c|c}\mathcal{L}[\alpha,\alpha]&\mathcal{L}[\alpha,\beta]&\mathcal{L}[\alpha,\gamma]\\\hline \mathbf{0}_{|\beta| \times |\alpha|}&\mathcal{L}[\beta,\beta]&\mathbf{0}_{|\beta| \times |\gamma|}\\\hline \mathcal{L}[\gamma,\alpha]&\mathcal{L}[\gamma,\beta]&\mathcal{L}[\gamma,\gamma]\end{array}\right],
    \end{align}
    where $\gamma = (\alpha \cup \beta)^c$.  $\mathcal{L}[\beta,\beta] \mathbf{1}_{|\beta|} = \mathbf{0}_{|\beta|}$ follows from the definition of loop-less Laplacian and (\ref{unreachable}). Because $\det \mathcal{L}[\alpha^c,\alpha^c] = \det \mathcal{L}[\beta,\beta]  \det \mathcal{L}[\gamma, \gamma]$ and $\mathcal{L}[\beta,\beta]$ has a zero eigenvalue, the singularity of $\mathcal{L}[\alpha^c,\alpha^c]$ is guaranteed and the proof is complete.
\end{proof}

The following lemma, which is a generalization of Lemma II.1 in \cite{dorfler2012kron},
shows the structural properties of the Kron reduction method.
\begin{lemma}[\textit{Structural Properties of Kron Reduction}]\label{lem: Structure}
    Let $\G = ([n], \E, \A)$ be a directed and weighted graph and $\alpha \subsetneq [n]$ be a reachable subset of nodes with $|\alpha|\ge 2$. The following statements hold:
    \begin{enumerate}  
        \item  The right accompanying matrix $\Q_{\mathrm{r a c}} := -\Q[\alpha^c,\alpha^c]^{-1}\Q[\alpha^c,\alpha] \in \R^{(n - |\alpha|)\times |\alpha| }$
        and the left accompanying matrix $\Q_{\mathrm{l a c}}:= -\Q[\alpha,\alpha^c]\Q[\alpha^c,\alpha^c]^{-1}\in \R^{|\alpha| \times (n - |\alpha|)}$ are non-negative. If the
        subgraph among the interior nodes is strongly connected and each boundary node $j \in \alpha$ is adjacent from (to) at least one interior node $i \in \alpha^c$, i.e., $(i,j) \in \E$ ($(j,i) \in \E$) then $\Q_{\mathrm{r a c}}$ ($\Q_{\mathrm{l a c}}$) is positive.
         In addition, if $\Q \equiv \mathcal{L}$ is a loop-less Laplacian, then $\Q_{\mathrm{r a c}} = \mathcal{L}_{\mathrm{r a c}} := -\mathcal{L}[\alpha^c, \alpha^c]^{-1} \mathcal{L}[\alpha^c, \alpha]$ is row stochastic.
         \item \textbf{Closure properties:} If $\Q$ is a loopy Laplacian, then $\Q_{\mathrm{red}}$ is also a loopy Laplacian matrix. Moreover, if $\Q$ is loop-less, then $\Q_{\mathrm{red}}$ is loop-less.
    \end{enumerate}
\end{lemma}
\begin{proof}
The proof of $1)$ is similar to Lemma I\hspace{-.1em}I.1 of \cite{dorfler2012kron}.
Note that if $\G$ is directed, then $\Q_{\mathrm{lac}}$ is not always equivalent to $\Q_{\mathrm{rac}}^{\top}$.
To prove statement $2)$, note that the difference from Lemma I\hspace{-.1em}I.1 of \cite{dorfler2012kron} is that $\Q$ is not always an M-matrix. However, since $\Q[\alpha, \alpha^c]$ is non-positive, $\Q_{\mathrm{red}} = \Q[\alpha,\alpha] + \Q[\alpha,\alpha^c]\Q_{\mathrm{rac}}$ is a Z-matrix. The closure property of the Schur complement \cite{zhang2006schur} states that $\Q_{\mathrm{red}}$ is row diagonally dominant. Hence, $\Q_{\mathrm{red}}$ is a loopy Laplacian.
The proof of the closure of the loop-less case follows from the same argument in the proof of Lemma I\hspace{-.1em}I.1. of \cite{dorfler2012kron}.
\end{proof}

\textcolor{black}{The consequence of Lemma  \ref{lem: Structure} is that $\Q_{\mathrm{red}}$ is a loopy Laplacian that induces again a directed and weighted graph, which we denote $\G_{\mathrm{red}} = (\alpha, \E_{\mathrm{red}}, \A_{\mathrm{red}})$.}

\begin{example}[Electrical networks]\label{ex: Electrical}
We consider the connected and resistive electrical network model with $n$ nodes, current injections $I \in \R^{n}$, nodal voltages $V \in \R^{n}$, branch conductances $A_{i j} \ge 0$, and shunt conductances $A_{i i} \ge 0$ connecting node $i$ to the ground. The current-balance equations are $I = Q V$, where the conductance matrix $Q \in \R^{n \times n}$ is a symmetric loopy Laplacian. When we partition the nodes into two sets $[n] = \alpha \cup \alpha^c$, the associated partitioned current-balance equations are
\begin{align*}
    \left[\begin{array}{c}
         I[\alpha]  \\ \hline
         I[\alpha^c]
    \end{array}\right] &= \left[\begin{array}{c|c}
         Q[\alpha, \alpha] & Q[\alpha, \alpha^c] \\ \hline
         Q[\alpha^c, \alpha] & Q[\alpha^c, \alpha^c]
    \end{array}\right]
    \left[\begin{array}{c}
         V[\alpha]  \\ \hline
         V[\alpha^c]
    \end{array}\right].
\end{align*}
By using a Gaussian elimination method, we obtain the reduced model
\begin{align*}
    I_{\mathrm{red}} := I[\alpha] + Q_{\mathrm{ac}} I[\alpha^c] = Q_{\mathrm{red}} V[\alpha],
\end{align*}
where $I_{\mathrm{red}}$ and $Q_{\mathrm{red}}$ can be interpreted as a reduced current injection and a reduced conductance matrix, respectively, and the accompanying matrix $Q_{\mathrm{ac}} = Q_{\mathrm{l a c}} = Q_{\mathrm{r a c}}^{\top}$ maps the internal currents to the boundary currents in the reduced network. If the current injections are balanced, i.e., $\mathbf{1}_{n}^{\top} I = 0$ and $Q \equiv L$ has no shunt conductance, then $I_{\mathrm{red}}$ are also balanced, i.e., $\mathbf{1}_{|\alpha|}^{\top}I_{\mathrm{red}} = 0$, based on Lemma \ref{lem: Structure}.
\end{example}

The interior nodes can also be eliminated through a Kron reduction with multiple steps.
\begin{definition}[Iterative Kron Reduction]\label{def: Iterative Kron Reduction}
    Iterative Kron reduction associates to a loopy Laplacian matrix $\Q \in \R^{n \times n}$ and indices $\{1,\ldots, |\alpha|\}$, a sequence of matrices $\Q^l \in \R^{(n-l) \times (n-l)}, l \in \{1,\ldots, n-|\alpha|\}$, defined as
    \begin{align}
        \Q^l &= \Q^{l-1} / \Q^{l-1}[\{k_l\}, \{k_l\}] \label{Eq: Iterative KR}
    \end{align}
    where $\Q^0 = \Q$ and $k_l = n + 1 - l$.
\end{definition}

If $\alpha \subset \beta \subsetneq [n]$ and $\alpha$ is a reachable subset, $\beta$ is clearly also a reachable subset. 
Thus, the following lemma, which is a generalization of Lemma III.3 in \cite{dorfler2012kron},
follows directly from the quotient formula \cite{zhang2006schur} and Lemma \ref{lem: Structure}.
\begin{lemma}[\textit{Properties of  Iterative Kron Reduction}]\label{lem: Properties of Iterative KR}
    Let $\G = ([n], \E, \A)$ be a directed, and weighted graph and $\alpha \subsetneq [n]$ be a reachable subset with $|\alpha| \ge 2$. Consider the matrix sequence $\{\Q^l\}_{l=1}^{n-|\alpha|}$ defined through an iterative Kron reduction in Equation (\ref{Eq: Iterative KR}). The following statements hold:
    \begin{enumerate}
        \item \textbf{Well-posedness:} Each matrix $\Q^l$ is well-defined, and the classes of loopy and loop-less Laplacian matrices are closed throughout the iterative Kron reduction.
        \item \textbf{Quotient property:} The Kron-reduced matrix $\Q_{\mathrm{red}} = \Q/\Q[\alpha^c,\alpha^c]$ can be obtained thruough an iterative reduction of all interior nodes $k_l = n + 1 - l \in \alpha^c$, that is, $\Q_{\mathrm{red}} = \Q^{n - |\alpha|}$.
    \end{enumerate}
\end{lemma}
\begin{proof}
    We prove statement $1)$ based on an induction of $l \in \{1, \ldots, n-|\alpha|\}$.
    Since $[n-1] \supset \alpha$ is a reachable subset of $\G$, $\Q^{1} = \Q / \Q[\{n\}, \{n\}]$ is well-defined.
    Suppose $\Q^{l}$ $(1 \leq l \leq n-|\alpha|-1)$ is well-defined and $\Q^l = \Q / \Q[[n-l]^c, [n-l]^c]$. 
    Since $[n-l-1] \supset \alpha$ is a reachable subset of $\G$, $\Q / \Q[[n-l-1]^c, [n-l-1]^c]$ is well-defined.
    According to the quotient formula, $\Q[[n-l-1]^c, [n-l-1]^c] / \Q[[n-l]^c, [n-l]^c]$ is a nonsingular principal submatrix of $\Q / \Q[[n-l]^c, [n-l]^c]$, and
    \begin{align*}
        & \Q / \Q[[n-l-1]^c, [n-l-1]^c] \\
        &= (\Q / \Q[[n-l]^c, [n-l]^c]) / (\Q[[n-l-1]^c, [n-l-1]^c] / \Q[[n-l]^c, [n-l]^c])\\
        &= \Q^l / \Q^l[\{n-l\}, \{n-l\}]\\
        &= \Q^{l+1},
    \end{align*}
    where we used the induction hypothesis $\Q^l = \Q / \Q[[n-l]^c, [n-l]^c]$.
    Thus, $\Q^{l+1}$ is well-defined and $\Q^{l+1} = \Q / \Q[[n-l-1]^c, [n-l-1]^c]$.
    Therefore, $\Q^l$ is well-defined. The latter statement follows from Lemma \ref{lem: Structure}. Hence, statement $1)$ also follows.
    In particular, when we take $l$ as $n-|\alpha|$, then we obtain $\Q^{n-|\alpha|} = \Q / \Q[\alpha^c, \alpha^c] = \Q_{\mathrm{red}}$.
    Thus, statement $2)$ follows.
\end{proof}
\subsection{Algebraic and topological properties}
\label{subsec: Algebraic and topological properties}
In this subsection, we investigate the algebraic properties of a Kron reduction of a directed graph.
\begin{theorem}[\textit{Algebraic Properties of Kron Reduction}]\label{thm: Algebraic properties}
    Let $\G = ([n], \E, \A)$ be a directed and weighted graph and $\alpha \subsetneq [n]$ be a reachable subset with $|\alpha| \ge 2$. The following statements hold:
    \begin{enumerate}
        \item \textbf{Monotonic increase of weights:} For all $i, j \in \alpha$, it holds that $\A_{\mathrm{red}, i j} \geq \A_{i j}$. Correspondingly, it holds that $\Q_{\mathrm{red}, i j} \leq \A_{i j}$ for all distinct $i, j \in \alpha$.
        \item \textbf{Effect of self-loops:} Loopy and loop-less Laplacians $\Q$ and $\mathcal{L}$ are related by $\Q = \mathcal{L} + \mathrm{diag}(\{\A_{i i}\}_{i \in [n]})$. The Kron-reduced matrix then takes the following form:
        \begin{align*}
            \Q_{\mathrm{red}} &= \mathcal{L}/\mathcal{L}[\alpha^c, \alpha^c] + \mathrm{diag}(\{\A_{i i}\}_{i \in \alpha}) + \mathcal{S},
        \end{align*}
        where $\mathcal{S} = \mathcal{L}_{\mathrm{l a c}}(I_{n - |\alpha|} + \mathrm{diag}(\{\A_{i i}\}_{i \in \alpha^c})\mathcal{L}[\alpha^c, \alpha^c]^{-1})^{-1} \mathrm{diag}(\{\A_{i i}\}_{i \in \alpha^c})\mathcal{L}_{\mathrm{r a c}}$ is a non-negative matrix.
        Furthermore,
        \begin{align*}
            \A_{\mathrm{red}, i i} &= \A_{i i} + \sum_{j=1}^{n-|\alpha|} \Q_{\mathrm{lac}, i j} \A_{|\alpha|+j, |\alpha|+j}.
        \end{align*}
    \end{enumerate}
\end{theorem}
The proof of these statements follows analogously to the proof of Theorem I\hspace{-.1em}I\hspace{-.1em}I.6 of \cite{dorfler2012kron}. Note that if $\alpha = [n-1]$, then for all distinct $i, j \in [n-1]$,
\begin{align}
    \A_{\mathrm{red}, i j}&= \A_{i j} + \frac{\A_{i n}\A_{n j}}{\Q_{n n}} \geq \A_{i j}\label{Eq: off-diagonal},\\
        \A_{\mathrm{red}, i i}&= \A_{i i} + \frac{\A_{i n} \A_{n n}}{\Q_{n n}} \geq \A_{i i}\label{Eq: diagonal}.
\end{align}

The following theorem, which is a generalization to Theorem III.4 of \cite{dorfler2012kron}, shows how the graph topology of $\G$ changes under a Kron reduction, which follows directly from the equations (\ref{Eq: off-diagonal}), (\ref{Eq: diagonal}), and Lemma \ref{lem: Properties of Iterative KR}.
\begin{theorem}[\textit{Topological Properties of Kron Reduction}]\label{thm: topological}
    Let $\G = ([n], \E, \A)$, $\G_{\mathrm{red}} = (\alpha, \E_{\mathrm{red}}, \A_{\mathrm{red}})$ be directed and weighted graphs associated with $\Q, \Q_{\mathrm{red}} = \Q/\Q[\alpha^c, \alpha^c]$, respectively, where $\alpha$ is a reachable subset with $|\alpha| \ge 2$. The following statements then hold:
    \begin{enumerate}
        \item \textbf{Edges:} There is an edge from $i \in \alpha$ to $j \in \alpha$ in $\G_{\mathrm{red}}$ if and only if there is a path from $i$ to $j$ in $\G$ whose nodes all belong to $\{i, j\} \cup \alpha^c$. In particular, if $\G$ is strongly connected, $\G_{\mathrm{red}}$ is also strongly connected. 
        \item \textbf{Self-loops:} A node $j \in \alpha$ features a self-loop in $\G_{\mathrm{red}}$ if and only if $j$ features a self-loop in $\G$ or there is a path from a loopy interior node $i \in \alpha^c$ to $j$ whose nodes all belong to $\{i, j\} \cup \alpha^c$.
    \end{enumerate}
\end{theorem}
\begin{figure}[hbtp]
 \centering
 \includegraphics[keepaspectratio, scale=0.30]{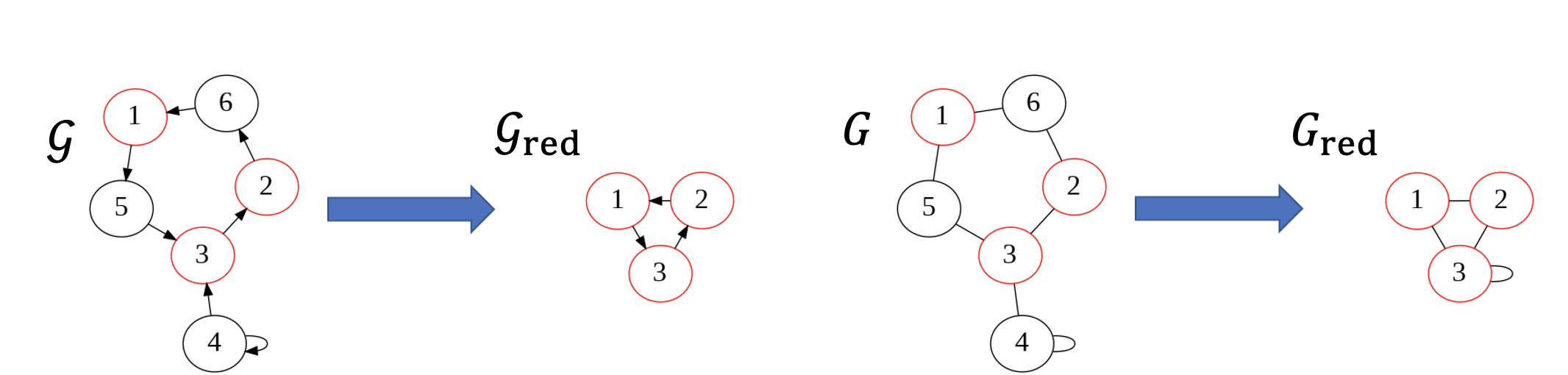}
 \caption{Illustration of Kron reduction of a directed graph (left) and the corresponding undirected graph (right) with three boundary nodes \textcolor{black}{$\{1,2,3\}$} and three interior nodes \textcolor{black}{$\{4,5,6\}$} (the edge weights were omitted for simplicity).}
 \label{fig: preservation of a self loop}
\end{figure}
\begin{example}
    Consider a directed graph $\G$ and corresponding undirected graph $G$ with six vertices and three boundary nodes \textcolor{black}{$\{1,2,3\}$}, as shown in Fig. \ref{fig: preservation of a self loop}.
    The loopy Laplacian matrices of $\G$ and $G$ are
    \begin{align*}
        \Q &= \begin{bmatrix}1&0&0&0&-1&0\\0&1&0&0&0&-1\\0&-1&1&0&0&0\\0&0&-1&2&0&0\\0&0&-1&0&1&0\\-1&0&0&0&0&1\end{bmatrix}
        \,\, \mathrm{and}\,\, Q = \begin{bmatrix}2&0&0&0&-1&-1\\0&2&-1&0&0&-1\\0&-1&3&-1&-1&0\\0&0&-1&2&0&0\\-1&0&-1&0&2&0\\-1&-1&0&0&0&2\end{bmatrix},
    \end{align*}
    respectively. Then, the Kron-reduced matrices of $\G_{\mathrm{red}}$ and $G_{\mathrm{red}}$ are given by
    \begin{align*}
        \Q_{\mathrm{red}}
        &= \begin{bmatrix}1&0&-1\\-1&1&0\\0&-1&1\end{bmatrix}\,\,
        \mathrm{and}\,\,Q_{\mathrm{red}} = 
        \begin{bmatrix}1&-\frac{1}{2}&-\frac{1}{2}\\-\frac{1}{2}&\frac{3}{2}&-1\\-\frac{1}{2}&-1&2\end{bmatrix},
    \end{align*}
    respectively.
     This example illustrates that even if the original directed graph has a self-loop, the reduced graph may feature no self-loops, whereas the Kron reduction of an undirected graph preserves the existence of a self-loop.
\end{example}

\subsection{Weight balanced directed graph}
\label{subsec: Weight balanced directed graph}
In this subsection, to study the properties of effective resistance, we introduce a \textit{weight balanced directed graph}, which is a class of directed graphs that will be used later in Section \ref{sec: Effective resistance of directed graphs}.
\begin{definition}\label{def: Weight balanced directed graph}
    A directed graph $\G = ([n], \E, \A)$ is said to be weight balanced if the in-degree and out-degree coincide for each node, i.e.,
    \begin{align*}
        \mathcal{D}_{i i} &= \sum_{j=1}^n \A_{i j} = \sum_{j=1}^n \A_{j i}, \,\, i \in [n].
    \end{align*}
\end{definition}
If $\mathcal{L} \in \R^{n \times n}$ is a loop-less Laplacian of a weight balanced directed graph, then
\begin{align*}
    \mathbf{1}_n^{\top} \mathcal{L} &= \left(\mathcal{D}_{1 1} - \sum_{j=1}^n \A_{j 1}, \ldots, \mathcal{D}_{n n} - \sum_{j=1}^n \A_{j n}\right) = \mathbf{0}_n^{\top}.
\end{align*}
Every undirected graph is clearly weight balanced.
Thus, the class of weight balanced directed graphs is a generalization of the class of undirected graphs, which plays an important role in the consensus coordination of multi-agent systems \cite{cortes2008distributed}.

The following theorem shows that a Kron reduction preserves the weight balance.
\begin{theorem}\label{thm: preservation of balanced}
    Let $\G = ([n], \E, \A)$ be a weight balanced directed graph and $\alpha \subsetneq [n]$ be a \textcolor{black}{reachable} subset with $|\alpha| \ge 2$.
    Suppose $\A_{i i} = 0$ for all $i \in \alpha^c$.
    Then, the Kron-reduced graph $\G_{\mathrm{red}} = (\alpha, \E_{\mathrm{red}}, \A_{\mathrm{red}})$ is also weight balanced.
\end{theorem}
\begin{proof}
    According to Lemma \ref{lem: Properties of Iterative KR}, it suffices to prove the theorem in the case when $\alpha = [n-1]$. For $i \in [n-1]$,
    \begin{align*}
        \sum_{j=1}^{n-1}\A_{\mathrm{red}, i j} 
        &= \sum_{j \in [n-1] \setminus \{i\}}\left(\A_{i j} + \frac{\A_{i n} \A_{n j}}{\Q_{n n}}\right) + \A_{i i} + \frac{\A_{i n}\A_{n n}}{\Q_{n n}}\\
        &= \mathcal{D}_{i i} - \frac{\A_{i n} \A_{n i}}{\mathcal{D}_{n n}}\\
        &= \sum_{j \in [n-1] \setminus \{i\}}\left(\A_{j i} + \frac{\A_{j n} \A_{n i}}{\Q_{n n}}\right) + \A_{i i} + \frac{\A_{i n}\A_{n n}}{\Q_{n n}}\\
        &= \sum_{j=1}^{n-1}\A_{\mathrm{red}, j i} 
    \end{align*}
    where we used equalities (\ref{Eq: off-diagonal}), (\ref{Eq: diagonal}), $\A_{n n} = 0$, and $\mathcal{Q}_{n n} = \mathcal{D}_{n n} = \sum_{j=1}^n \A_{i j} = \sum_{j=1}^n \A_{j i}$.
\end{proof}
\section{Effective resistance of directed graphs}
\label{sec: Effective resistance of directed graphs}
In this section, drawing inspiration from the relationship between electrical networks and Markov chains, we propose a generalized definition of the effective resistance applied to a directed graph.
\subsection{Undirected case}
\label{subsec: Undirected case}
In this subsection, we recall basic facts regarding the effective resistance of an undirected graph.
\begin{definition}(see  \cite{dorfler2012kron} \cite{klein1993resistance})\label{def: ER of undirected}
        Let $G = ([n], E, A)$ be an undirected, connected, and weighted graph, and let $Q$ be the corresponding loopy Laplacian.
        Then, the effective resistance $R_{G}(a, b)$ and the effective conductance $C_{G}(a, b)$ between two nodes $a, b \in [n]$ are
        \begin{align*}\label{Eq: def_of_EL_of_UD}
            R_{G}(a, b) &:= (e_a - e_b)^{\top}Q^{\dagger}(e_a - e_b) = Q^{\dagger}_{a a} + Q^{\dagger}_{b b} - Q^{\dagger}_{a b}- Q^{\dagger}_{b a},\\
            C_{G}(a, b) &:= \frac{1}{R_{G}(a, b)},
        \end{align*}
        where $Q^{\dagger}$ is the Moore-Penrose pseudoinverse of $Q$.
\end{definition}

\textcolor{black}{These concepts are derived from the representation of a graph as a network of resistors, such as in Example \ref{ex: Electrical}.
The effective conductance $C_{G}(a, b)$ of a loop-less graph is calculated as the amount of flow from nodes $a$ to $b$ when we fix their potential $V$ to $1$ and $0$, respectively, which is mentioned in \cite{loebl2010discrete}.
Here, the nodes $a$ and $b$ are the only vertices at which the current can leave or enter the circuit.
The effective resistance $R_G(a,b)$ is defined as $R_G(a,b)=1/C_G(a,b)$.
In this case, the current-balance equations are $C_{G}(a, b)(e_a - e_b) = Q V$. In particular,
\begin{align}
    C_G(a,b) &= (Q V)_a = \sum_{x=1}^n Q_{ax}V_x = \sum_{x=1}^n A_{a x}(V_a - V_x)\label{Eq: EC}.
\end{align}
If $G$ is loop-less, then the general solution to the current-balance equations is $V = C_G(a,b) Q^{\dagger}(e_a - e_b) + c\mathbf{1}_n$, where $c \in \R$ is an arbitrary constant. Then,
\begin{align*}
    R_G(a,b) &= R_G(a,b)(V_a - V_b)\\
    &= R_G(a,b)(e_a-e_b)^{\top}V \\
    &= R_G(a,b)(e_a - e_b)^{\top}(C_G(a,b) Q^{\dagger}(e_a - e_b)+c\mathbf{1}_n)\\
    &= (e_a - e_b)^{\top}Q^{\dagger}(e_a - e_b),
\end{align*}
where we used the assumption that $V_a=1$ and $V_b = 0$. Thus, we recover Definition \ref{def: ER of undirected}.
}

As shown in Theorem I\hspace{-.1em}I\hspace{-.1em}I.8 of \cite{dorfler2012kron}, the effective resistances between boundary nodes are invariant under a Kron reduction of the interior nodes.
\subsection{Novel definition of effective resistance}
In this subsection, we define the effective conductance of directed graphs based on the physical interpretation.

Let $\G = ([n], \E, \A)$ be a directed graph and fix distinct $a, b \in [n]$. 
To define the effective conductance $C_{\G}(a, b)$, we assume the following assumptions:
\begin{itemize}
    \item $\{a, b\} \subset [n]$ is a reachable subset.
    \item $\mathcal{D}_{i i} > 0$ for $i \in \{a, b\}$. If there are no outgoing edges from $i$, then we add a virtual self-loop $\A_{i i} = 1$.
\end{itemize}
Under these assumptions, $\mathcal{D}_{i i} > 0$ for all $i \in [n]$.
Then, we define the corresponding transition matrix using $P := \mathcal{D}^{-1}\A$, which induces a Markov chain $X = \{X_m\}_{m=0}^{\infty}$ with the probability transition matrix $P$ on $[n]$, that is,
\begin{align*}
    \mathrm{Prob}(X_{m+1} = j | X_{m} = i) &= P_{i j} = \frac{\A_{i j}}{\mathcal{D}_{i i}}
\end{align*}
for all $i,j \in [n]$ and $m \in \N$.
Moreover, define $V_{a\to b}(x) := P^{x}(\sigma_a < \sigma_b)$ for all $x \in [n]$, where $P^x$ represents the probability law of $X$ given $X_0 = x$ and $\sigma_a$ indicates the hitting time of $a \in [n]$ by
\begin{align*}
    \sigma_{a} &:= \inf\{m \geq 0 : X_{m} = a\} \in [0, \infty].
\end{align*}
In other words, $V_{a \to b}(x)$ is the probability that a chain started at $x$ reaches $a$ before reaching $b$. Clearly, $V_{a \to b}(a) = 1, V_{a \to b}(b) = 0$. If $x$ is neither $a$ nor $b$,
based on the Markov property, we obtain
\begin{align}
    V_{a \to b}(x) &= \sum_{y=1}^n P_{x y} V_{a \to b}(y) = \sum_{y \in \{a, b\}^c}P_{x y}V_{a \to b}(y) + P_{x a}.\label{Eq: Markov property}
\end{align}
Note that if $\G$ is undirected and loop-less, then $V_{a \to b}(x)$ is the voltage $V_x$ when the unit voltage with $V_a = 1$ and $V_b = 0$ is applied between $a$ and $b$ \cite{doyle1984random}.
In this sense, $V_{a \to b}(x)$ is an extension of the voltage concept. 
\begin{lemma}\label{lem: well-definedness of V}
    $V_{a \to b} = (V_{a \to b}(1), \ldots, V_{a \to b}(n))^{\top} \in \R^n$ is uniquely determined.
\end{lemma}
\begin{proof}
    It suffices to prove the lemma in the case when $a = 1, b = 2$ and $n \ge 3$.
    By definition, $V_{1 \to 2}(1) = 1$ and $V_{1 \to 2}(2) = 0$.
     Eq. (\ref{Eq: Markov property}) can be rewritten as
    \begin{align*}
        \begin{bmatrix}1-P_{3 3}&-P_{3 4}&\ldots&-P_{3 n}\\-P_{4 3}&1-P_{4 4}&\ldots&-P_{4 n}\\\vdots&\vdots&\ddots&\vdots\\-P_{n 3}&-P_{n 4}&\ldots&1-P_{n n}\end{bmatrix}
        \begin{bmatrix}V_{1 \to 2}(3)\\V_{1 \to 2}(4)\\\vdots\\V_{1 \to 2}(n)\end{bmatrix}
        &= \begin{bmatrix}P_{3 1}\\P_{4 1}\\\vdots\\P_{n 1}\end{bmatrix}.
    \end{align*}
    By multiplying each side by $\mathcal{D}[\{1,2\}^c, \{1,2\}^c]$, we obtain
    \begin{align}
        \begin{bmatrix}\mathcal{D}_{3 3}-\A_{3 3}&-\A_{3 4}&\ldots&-\A_{3 n}\\-\A_{4 3}&\mathcal{D}_{4 4}-\A_{4 4}&\ldots&-\A_{4 n}\\\vdots&\vdots&\ddots&\vdots\\-\A_{n 3}&-\A_{n 4}&\ldots&\mathcal{D}_{n n}-\A_{n n}\end{bmatrix}
        \begin{bmatrix}V_{1 \to 2}(3)\\V_{1 \to 2}(4)\\\vdots\\V_{1 \to 2}(n)\end{bmatrix}
        &= \begin{bmatrix}\A_{3 1}\\\A_{4 1}\\\vdots\\\A_{n 1}\end{bmatrix}.\label{Eq: Due to strictly loopy case}
    \end{align}
    The matrix of (\ref{Eq: Due to strictly loopy case}) is clearly $\mathcal{L}[\{1,2\}^c, \{1,2\}^c]$, which is nonsingular according to Lemma \ref{lem: Well-definedness of Kron Reduction}.
    Thus, $V_{1 \to 2}[\{1,2\}^c] = \mathcal{L}[\{1,2\}^c, \{1,2\}^c]^{-1}\A[\{1, 2\}^{c}, \{1\}]$ is uniquely determined, thereby completing the proof.
\end{proof}

To define the effective conductance, we replace $V$ and $A$ in (\ref{Eq: EC}) with $V_{a \to b}$ and $\A$:
\begin{align}
    \sum_{x=1}^{n}\A_{a x}(V_{a \to b}(a) - V_{a \to b}(x))
    &= \mathcal{D}_{a a} - \sum_{x=1}^{n}\A_{a x}V_{a \to b}(x)\notag\\
    &= \mathcal{D}_{a a} - \A_{a a}V_{a \to b}(a) - \A_{a b}V_{a \to b}(b) -  \sum_{x \in \{a,b\}^c}\A_{a x}V_{a \to b}(x)\notag\\
    &= (\mathcal{D}_{a a} - \A_{a a})  + \sum_{x \in \{a,b\}^c}\mathcal{L}_{a x}V_{a \to b}(x)\notag\\
    &= \mathcal{L}_{a a} + \mathcal{L}[\{a\}, \{a,b\}^c] V_{a \to b}[\{a, b\}^c],\label{Eq: calculate EC part 1}
\end{align}
where we used the relations $V_{a \to b}(a) = 1, V_{a \to b}(b) = 0$, $\mathcal{L}_{i j} = -\A_{i j}$ and $\mathcal{L}_{i i} = \mathcal{D}_{i i} - \A_{i i}$ for all distinct $i, j \in [n]$.
According to the proof of Lemma \ref{lem: well-definedness of V},
\begin{align*}
    V_{a \to b}[\{a,b\}^c] &= \mathcal{L}[\{a,b\}^c, \{a,b\}^c]^{-1}\A[\{a, b\}^{c}, \{a\}] \\
    &= -\mathcal{L}[\{a,b\}^c, \{a,b\}^c]^{-1}\mathcal{L}[\{a, b\}^{c}, \{a\}].
\end{align*}
Thus, 
\begin{align}
    &\mathcal{L}_{a a} + \mathcal{L}[\{a\}, \{a,b\}^c] V_{a \to b}[\{a,b\}^c]\notag \\
    =& \mathcal{L}_{a a} - \mathcal{L}[\{a\}, \{a,b\}^c](\mathcal{L}[\{a,b\}^c, \{a,b\}^c])^{-1}\mathcal{L}[\{a,b\}^c, \{a\}]\notag \\
    =:& \left(\mathcal{L}/\mathcal{L}[\{a,b\}^c,\{a,b\}^c]\right)_{i_a i_a},\label{Eq: calculate EC part 2}
\end{align}
where
$i_a \in \{1,2\}$ is the index corresponding to $a$ in $\mathcal{L}$.
Because this expression is unaffected by self-loops, including $\A_{a a}$ and $\A_{b b}$,
the effective conductance between two nodes in a directed graph is defined as follows:
\begin{definition}\label{def: EC and ER of directed graphs}
    Let $\G = ([n], \E, \A)$ be a directed graph and $\{a,b\} \subset [n]$ be a reachable subset. 
    Then, the effective conductance and effective resistance from node $a$ to $b$ in $\G$ are defined as
    \begin{align}
        C_{\G}(a, b) &:= \left(\mathcal{L}/\mathcal{L}[\{a,b\}^c,\{a,b\}^c]\right)_{i_a i_a}\label{Eq: Def of EC},\\
        R_{\G}(a, b) &:= \frac{1}{C_{\G}(a, b)},\label{Eq: Def of ER}
    \end{align}
    respectively, where $\mathcal{L}$ is a loop-less Laplacian corresponding to $\G$.
    If $C_{\G}(a,b) = 0$, we define $R_{\G}(a,b) = \infty$.
    By an abuse notation, we denote $i_a$ by $a$.
\end{definition}
Note that the above effective resistance differs from that
introduced in \cite{young2015new}, as described in Section \ref{subsubsec: ER based on Lyapunov theory}.

The following lemma shows the relationship between the effective conductance and Kron reduction.
\begin{lemma}(\textit{Resistive Properties of Kron Reduction})\label{lem: Resistive Properties of KR}
    Let $\G = ([n], \E, \A)$ be a directed graph and $\{a,b\} \subset [n]$ be a reachable subset. 
    The following statements then hold:
    \begin{enumerate}
        \item The Kron-reduced matrix $\mathcal{L}/\mathcal{L}[\{a,b\}^c,\{a,b\}^c]$ takes the form
        \begin{align}
            \mathcal{L}/\mathcal{L}[\{a,b\}^c,\{a,b\}^c] &= \begin{bmatrix}C_{\G}(a, b)&-C_{\G}(a, b)\\-C_{\G}(b, a)&C_{\G}(b, a)\end{bmatrix}. \label{Eq: relationship between ER and KR}
        \end{align}
        \item If $\G$ is loop-less, the effective conductance and the effective resistance are invariant under a Kron reduction.
    \end{enumerate}
\end{lemma}
\begin{proof}
    Since $\{a, b\} \subset [n]$ is a reachable subset, Lemma \ref{lem: Structure} implies that $\mathcal{L}/\mathcal{L}[\{a,b\}^c,\{a,b\}^c]$ is a loop-less Laplacian.
    Thus,
    \begin{align*}
        \begin{bmatrix}\left(\mathcal{L}/\mathcal{L}[\{a,b\}^c,\{a,b\}^c]\right)_{a a}&\left(\mathcal{L}/\mathcal{L}[\{a,b\}^c,\{a,b\}^c]\right)_{a b}\\\left(\mathcal{L}/\mathcal{L}[\{a,b\}^c,\{a,b\}^c]\right)_{b a}&\left(\mathcal{L}/\mathcal{L}[\{a,b\}^c,\{a,b\}^c]\right)_{b b}\end{bmatrix}
        \begin{bmatrix}1\\1\end{bmatrix} &=
        \begin{bmatrix}
            0\\0
        \end{bmatrix}.
    \end{align*}
     We therefore obtain
    \begin{align*}
        \begin{cases}
            \left(\mathcal{L}/\mathcal{L}[\{a,b\}^c,\{a,b\}^c]\right)_{a b} &= -\left(\mathcal{L}/\mathcal{L}[\{a,b\}^c,\{a,b\}^c]\right)_{a a}= -C_{\G}(a, b), \\
            \left(\mathcal{L}/\mathcal{L}[\{a,b\}^c,\{a,b\}^c]\right)_{b a} &= -\left(\mathcal{L}/\mathcal{L}[\{a,b\}^c,\{a,b\}^c]\right)_{b b}\,= -C_{\G}(b, a),
        \end{cases}
    \end{align*}
    which proves the identity (\ref{Eq: relationship between ER and KR}).
    
    To prove statement $2)$, fix the subset $\alpha \subseteq [n]$ containing $\{a, b\} \subset \alpha$ and denote $\mathcal{L}_{\mathrm{red}} = \mathcal{L}/\mathcal{L}[\alpha^c\textcolor{black}{,} \alpha^c]$.
    According to the quotient formula \cite{zhang2006schur}, $\mathcal{L}_{\mathrm{red}}[\{a, b\}^c, \{a, b\}^c]$ is a nonsingular principal submatrix of $\mathcal{L}_{\mathrm{red}}$, and
    \begin{align*}
        \begin{bmatrix}C_{\G}(a, b)&-C_{\G}(a, b)\\-C_{\G}(b, a)&C_{\G}(b, a)\end{bmatrix} &= \mathcal{L}/\mathcal{L}[\{a,b\}^c,\{a,b\}^c]\\
        &=  \mathcal{L}_{\mathrm{red}}/\mathcal{L}_{\mathrm{red}}[\{a,b\}^c, \{a,b\}^c] \\
        &= \begin{bmatrix}C_{\G_{\mathrm{red}}}(a, b)&-C_{\G_{\mathrm{red}}}(a, b)\\-C_{\G_{\mathrm{red}}}(b, a)&C_{\G_{\mathrm{red}}}(b, a)\end{bmatrix},
    \end{align*}
    for which we used statement $1)$.
    Thus, statement $2)$ follows.
\end{proof}
\begin{remark}[Strictly loopy case]
    If we consider the self-loops $\A_{i i}$ as \textit{shunt conductances} connecting node $i$ to the ground, the role can be better understood by introducing the
    additional \textit{grounded node} with index $n + 1$. In other words, we consider the augmented graph $\hat{\G}$ with the node set $[n+1]$, and replace the self-loop $\A_{i i}$ in $\G$ with a bidirectional edge between $i$ and $n+1$, whose weight is equal to $\A_{i i}$ in $\hat{\G}$.
    Because $\hat{\G}$ has no self-loops, the effective conductance of this interpretation can be defined by the corresponding augmented Laplacian matrix 
    \begin{align*}
        \hat{\Q} &:= \left[\begin{array}{c|c}\Q&-\mathrm{diag}(\{\A_{i i}\}_{i=1}^{n})\mathbf{1}_n\\\hline -\mathbf{1}_n^{\top} \mathrm{diag}(\{\A_{i i}\}_{i=1}^n)&\sum_{i=1}^n \A_{i i}\end{array}\right] \in \R^{(n+1) \times (n+1)}
    \end{align*}
    instead of the loop-less Laplacian $\mathcal{L} \in \R^{n \times n}$.
\end{remark}
\begin{figure}[hbtp]
    \centering
    \includegraphics[keepaspectratio, scale=0.35]{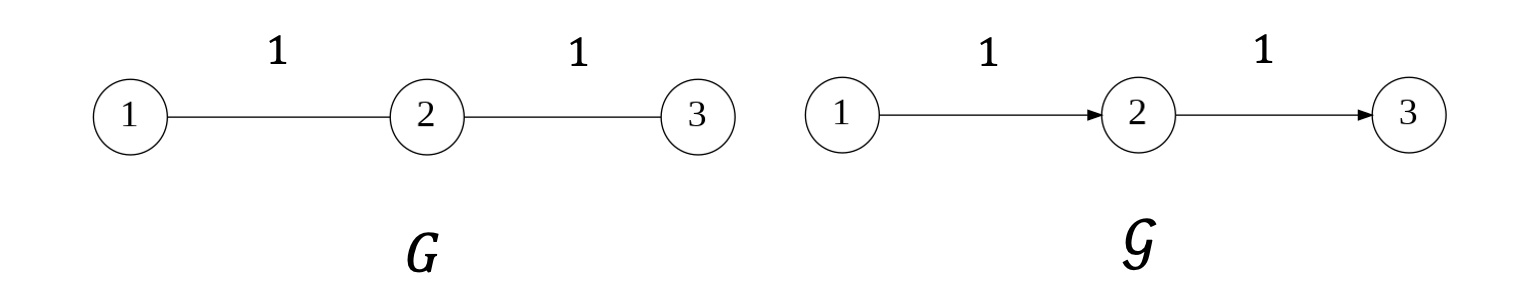}
    \caption{A simple 3-node directed graph $\G$ and the corresponding undirected graph $G$}.
    \label{fig: not_length}
\end{figure}
\begin{remark}
    It can be proven that the effective resistance is equivalent to the length of the shortest path when the graph is undirected, unweighted and acyclic.
    Unfortunately, this property does not hold in our generalization. Consider 3-node graphs, as shown in Fig. \ref{fig: not_length}.
    For undirected graph $G$, we can compute that $R_{G}(1, 3) = R_{G}(3, 1) = 2$. However, for directed graph $\G$, $R_{\G}(1, 3) = 1 \neq 2$ and $R_{\G}(3, 1) = \infty \neq 2$.
\end{remark}

We can interpret the effective conductance probabilistically, as shown in 
the following theorem, which is a generalization of the undirected case in \cite{doyle1984random} applied to a directed graph.
\begin{theorem}\label{lem: probabilistic interpretation}
        Let $\G = ([n], \E, \A)$ be a directed graph and $\{a, b\} \subset [n]$ be a reachable subset.
        Suppose $\mathcal{D}_{i i} > 0$ for all $i \in [n]$, and consider a Markov chain $X$ with the probability transition matrix $P = \mathcal{D}^{-1}\A$ on $[n]$.
        The following statements hold.
        \begin{enumerate}
            \item Let $P_{\mathrm{esc}}(a, b)$ be the probability, starting at $a$, that the walk reaches $b$ before returning to $a$. Then,
            \begin{align*}
                C_{\G}(a, b) = \mathcal{D}_{a a}P_{\mathrm{esc}}(a, b).
            \end{align*}
            \item If there is an edge from $a$ to $b$, the probability that the walk starting at $a$ reaches $b$ for the first time using the edge $\{a, b\} \in \E$ is equal to $\A_{a b}/C_{\G}(a, b)$.
        \end{enumerate}
    \end{theorem}
The probability $P_{\mathrm{esc}}(a, b)$ is said to be the escape probability from $a$ to $b$.
\begin{proof}
    It follows from (\ref{Eq: calculate EC part 1}),(\ref{Eq: calculate EC part 2}), and (\ref{Eq: Def of EC}) that
    \begin{align*}
        C_{\G}(a, b) = \left(\mathcal{L}/\mathcal{L}[\{a,b\}^c,\{a,b\}^c]\right)_{a a}
        = \mathcal{D}_{a a}\left(1 - \sum_{x=1}^n P_{a x}V_{a \to b}(x)\right).
    \end{align*}
    Recall that $V_{a \to b}(x)$ is the probability that the walk starting at $x$ reaches $a$ before $b$.
    Then, $\sum_{x=1}^n P_{a x}V_{a \to b}(x)$ is the probability that the walk starts at $a$ and returns to $a$ before reaching $b$, which is equal to $1 - P_{\mathrm{esc}}(a, b)$, thus proving statement $1)$.
    
    According to Theorem \ref{thm: Algebraic properties} and Lemma \ref{lem: Resistive Properties of KR},
    \begin{align*}
        C_{\G}(a, b) &= -\left(\mathcal{L}/\mathcal{L}[\{a,b\}^c,\{a,b\}^c]\right)_{a b} \geq -\mathcal{L}_{a b} = \A_{a b}.
    \end{align*}
    This inequality clearly shows that if there is an edge from $a$ to $b$, then $0 \leq \A_{a b}/ C_{\G}(a, b) \leq 1$.
    Let $p$ be the probability that a walk starting at $a$ reaches $b$ for the first time using edge $\{a,b\}$. 
    The probability that the walk starts with edge $\{a, b\}$ is $P_{a b}$.
    It also has a probability $1 - P_{\mathrm{esc}}(a, b)$ of not visiting $b$ before returning to $a$.
    In this case, the probability of visiting $b$ using edge $\{a,b\}$ is $p$. Therefore,
        $p = P_{a b} + (1 - P_{\mathrm{esc}}(a, b))p$,
    and thus
    \begin{align*}
        p = \frac{P_{a b}}{P_{\mathrm{esc}}(a, b)}
        = \frac{\A_{a b}}{C_{\G}(a, b)},
    \end{align*}
    which completes the proof of statement $2)$.
\end{proof}

\begin{remark}
    Without using the concept of the reachable subset, a generalization of Definition \ref{def: EC and ER of directed graphs}  is difficult to achieve for the following reasons.
    \begin{itemize}
        \item Lemma \ref{lem: Well-definedness of Kron Reduction} states that
    if the reachability assumption of $\{a,b\} \subset [n]$ is removed,
    $\mathcal{L}[\{a,b\}^c,\{a,b\}^c]$ is singular and the equations (\ref{Eq: Def of EC}) and (\ref{Eq: Def of ER}) are ill-defined.
    In other words, there is a probability that a walk starting from $a$ will not reach either $a$ or $b$ under a general Markov chain. In this case, the escape probability $P_{\mathrm{esc}}(a,b)$ cannot be defined.
    \item The conditions for when the effective resistance is defined between a pair of nodes is generally non-transitive.
    That is, even if $R_{\mathcal{G}}(a,b)$ and $R_{\mathcal{G}}(b,c)$ can be defined,
    $R_{\mathcal{G}}(a,c)$ generally cannot.
    In fact, consider again a directed graph, $\mathcal{G}$ in Fig. \ref{fig: not_transitive}. 
    Then, because $\{1,2\}$ and $\{2,3\}$ are reachable subsets, $R_{\mathcal{G}}(1,2)$ and $R_{\mathcal{G}}(2,3)$ can be defined.
    By contrast, $R_{\mathcal{G}}(1,3)$ cannot be defined, because $\{1,3\}$ is not reachable.
    \end{itemize}
\end{remark}

\subsection{Effective resistance of strongly connected graphs}
\label{subsec: ER of SC graphs}
In this subsection, we focus on strongly connected graphs.
In this case, any pair of distinct nodes constitute a reachable set. Moreover, for any pair of distinct nodes, the effective resistance is finite, as shown in the following.
\subsubsection{Weight balanced directed graphs}
\label{subsubsec: ER of Weight balanced directed graphs}
First, we consider the weight balanced directed graphs introduced in Subsection \ref{subsec: Weight balanced directed graph}.
The following theorem shows that the effective resistance of strongly connected weight balanced directed graphs can be calculated using a formula similar to that applied to undirected cases.
\begin{theorem}[\textit{Effective Resistance of Weight Balanced Directed Graph}]\label{thm: ER of Weight Balanced Directed Graph}
        Let $\G = ([n], \E, \A)$ be a strongly connected weight balanced directed graph and let $\mathcal{L}$ be the corresponding loop-less Laplacian matrix.
        Subsequently, the effective resistance from node $a$ to node $b$ in $\G$ is
        \begin{align}
            R_{\G}(a, b) &= (e_a - e_b)^{\top}\mathcal{L}^{\dagger}(e_a - e_b)\label{Eq: expression of EL of Balanced}\\
            &= (e_a - e_b)^{\top}\mathcal{L}^{\dagger}_s(e_a - e_b), \label{Eq: expression of EL of Balanced via sym}
        \end{align}
        where $\mathcal{L}^{\dagger}_s = (\mathcal{L}^{\dagger} + (\mathcal{L}^{\dagger})^{\top})/2$.
\end{theorem}
\begin{proof}
    According to Lemma 4 in \cite{fontan2021properties}, 
        $\mathcal{L}^{\dagger} = \left(\mathcal{L} + \frac{\gamma}{n}\mathbf{1}_{n \times n}\right)^{-1} - \frac{1}{n\gamma}\mathbf{1}_{n \times n}$
    for all $\gamma \neq 0$ if $\G$ is strongly connected weight balanced.
    Then, the right-hand side of (\ref{Eq: expression of EL of Balanced}) can be reformulated as
    \begin{align}
        (e_a - e_b)^{\top}\mathcal{L}^{\dagger}(e_a - e_b) &= (e_{i_a} - e_{i_b})^{\top}\left(\mathcal{L}/\mathcal{L}[\{a, b\}^c, \{a,b\}^c]\right)^{\dagger}(e_{i_a} - e_{i_b})\label{Eq: analogy},
    \end{align}
    which is analogous to the proof of Theorem I\hspace{-.1em}I\hspace{-.1em}I.8 in \cite{dorfler2012kron}, where $i_a, i_b \in \{1,2\}$ is the index corresponding to $a,b$ in $\mathcal{L}$ respectively.
    According to Theorem \ref{thm: preservation of balanced}, $\mathcal{L}/\mathcal{L}[\{a, b\}^c, \{a,b\}^c] \in \R^{2 \times 2}$ is weight balanced, and thus,
    \begin{align*}
        0 &= \left(\mathcal{L}/\mathcal{L}[\{a, b\}^c, \{a,b\}^c]\right)_{i_a, i_a} + \left(\mathcal{L}/\mathcal{L}[\{a, b\}^c, \{a,b\}^c]\right)_{i_a, i_b}\\
        & = \left(\mathcal{L}/\mathcal{L}[\{a, b\}^c, \{a,b\}^c]\right)_{i_a, i_a} + \left(\mathcal{L}/\mathcal{L}[\{a, b\}^c, \{a,b\}^c]\right)_{i_b\textcolor{black}{,} i_a}.
    \end{align*}
    Therefore,
    \begin{align*}
        \mathcal{L}/\mathcal{L}[\{a, b\}^c, \{a,b\}^c] &= \left(\mathcal{L}/\mathcal{L}[\{a, b\}^c, \{a,b\}^c]\right)_{i_a, i_a} \begin{bmatrix}1&-1\\-1&1\end{bmatrix} = \frac{1}{R_{\G}(a, b)}\begin{bmatrix}1&-1\\-1&1\end{bmatrix}.
    \end{align*}
    A direct computation shows that
        $\begin{bmatrix}x&-x\\-x&x\end{bmatrix}^{\dagger} = \frac{1}{4x}\begin{bmatrix}1&-1\\-1&1\end{bmatrix}$
    for all $x \neq 0$.
    By applying this matrix identity with $x =  1 / R_{\G}(a, b)$, (\ref{Eq: analogy}) is rendered as
    \begin{align*}
        (e_a - e_b)^{\top}\mathcal{L}^{\dagger}(e_a - e_b) 
        = \frac{R_{\G}(a, b)}{4}(e_{i_a} - e_{i_b})^{\top}\begin{bmatrix}1&-1\\-1&1\end{bmatrix}(e_{i_a} - e_{i_b})
        = R_{\G}(a, b),
    \end{align*}
    which proves the identity (\ref{Eq: expression of EL of Balanced}).\\
    \indent Equation (\ref{Eq: expression of EL of Balanced via sym}) follows from (\ref{Eq: expression of EL of Balanced}) and the identity $v^{\top}A v = v^{\top}((A + A^{\top})/2)v$ for $A \in \R^{n \times n}$ and $v \in \R^n$.
\end{proof}
\begin{remark}
    In \cite{fontan2021properties}, the authors defined the effective resistance of a signed digraph whose corresponding Laplacian $L$ is normal, i.e., $LL^{\top} = L^{\top}L$
and $-L$ is eventually exponentially positive, that is, there exists $t_0 \in \N$ such that $\mathrm{exp}(-Lt)$ is a positive matrix for all $t \ge t_0$. Moreover, they also mentioned
the extension to non-negative and strongly connected weight balanced directed graphs, which is given by (\ref{Eq: expression of EL of Balanced via sym}).
\end{remark}

The following theorem, which is a generalization of Theorem B in \cite{klein1993resistance}, shows that the effective resistance defines a metric on a weight balanced directed graph.
\begin{theorem}\label{thm: the metricity}
    The effective resistance from node $a$ to node $b$ of a strongly connected weight balanced directed graph with a loop-less Laplacian $\mathcal{L}$ is a metric.
    In other words, the following statements hold.
    \begin{enumerate}
        \item \textbf{Non-negativity:} $R_{\G}(a,b) \ge 0$. The equality holds if and only if $a = b$.
        \item \textbf{Symmetry:} $R_{\G}(a,b) = R_{\G}(b,a)$.\\
        \item \textbf{Triangle inequality:} $R_{\G}(a, b) \leq R_{\G}(a, c) + R_{\G}(c, b)$.
    \end{enumerate}
    Moreover, the effective resistance matrix $R_{\G} := (R_{\G}(a, b))_{a, b \in [n]} \in \R^{n \times n}$ is a \textit{Euclidean distance matrix},
     that is, there exists a set of $n$ vectors $x_1, \ldots, x_n \in \R^n$ such that 
     \begin{align*}
        (R_{\mathcal{G}})_{a b} &= \|x_a - x_b \|_2^2 = \sum_{i = 1}^n (x_{a i} - x_{b i})^2.
     \end{align*}
\end{theorem}
\begin{proof}
Statements
    $1), 2)$, and the fact that $R_{\G}$ is a Euclidean distance matrix are analogous to the proof of Lemma 5 in \cite{fontan2021properties}.
    \textcolor{black}{Here,
    we have used the fact that $\mathcal{L}^{\dagger}_s$ is a positive semidefinite of rank $n-1$, as shown in
    Theorem 4 in \cite{fontan2021properties}.}
    To prove the triangle inequality, we replace $\mathcal{L}$ with $\mathcal{L}/[\{a,b,c\}^c, \{a,b,c\}^c]$ based on a similar discussion as for (\ref{Eq: analogy}). Thus, it suffices to prove the theorem in the case when $a=1,b=2,c=3, n = 3$, and
    \begin{align*}
        \mathcal{L} &= \begin{bmatrix}
        \A_{12} + \A_{13} & -\A_{12} & -\A_{13}\\
        -\A_{21} & \A_{21} + \A_{23} & -\A_{23}\\
        -\A_{31} & -\A_{32} & \A_{31} + \A_{32}
        \end{bmatrix}.
    \end{align*}
    According to Definition \ref{def: EC and ER of directed graphs},
    \begin{align}
        R_{\G}(1,2) &= \frac{1}{(\mathcal{L}/\mathcal{L}[\{3\}, \{3\}])}_{11} = \frac{\A_{31} + \A_{32}}{(\A_{12} + \A_{13})(\A_{31}+\A_{32}) - \A_{13}\A_{31}}. \label{Eq: 3 times 3}
    \end{align}
    Since $\G$ is weight balanced, $\sum_{i \neq j}\A_{ij} = \sum_{i \neq j}\A_{j i}$ for all $j \in \{1,2,3\}$.
    Thus, the denominator on the right side of (\ref{Eq: 3 times 3})
    can be calculated as
    \begin{align*}
        (\A_{12} + \A_{13})(\A_{31}+\A_{32}) - \A_{13}\A_{31}
        &= (\A_{23} + \A_{21})(\A_{12}+\A_{13}) - \A_{21}\A_{12}\\
        &= (\A_{31} + \A_{32})(\A_{23}+\A_{21}) - \A_{32}\A_{23}.
    \end{align*}
    Therefore,
    \begin{align*}
        R_{\G}(1,2) + R_{\G}(2,3) - R_{\G}(1,3)
        &= \frac{(\A_{31} + \A_{32}) + (\A_{12} + \A_{13}) - (\A_{23} + \A_{21})}{(\A_{12} + \A_{13})(\A_{31}+\A_{32}) - \A_{13}\A_{31}}\\
        &= \frac{\A_{13} + \A_{31}}{\A_{12}\A_{31}+\A_{13}\A_{32} + \A_{12}\A_{32}} \ge 0.
    \end{align*}
    \textcolor{black}{Here, the last inequality follows from the fact that $\mathcal{A}_{ij}$ is nonnegative, as shown in Theorem \ref{thm: Algebraic properties}.}
\end{proof}

In \cite{klein1993resistance}, the authors proved the above theorem in an undirected case based on the use of an electrical network. However, these techniques cannot be applied to directed graphs.
Thus, the above proof is new.
\begin{remark}
    Let $R_{\G} \in \R^{n \times n}$ be the effective resistance of a strongly connected directed graph $\G$.
    Then, the total effective resistance \textcolor{black}{(or the Kirchhoff index)} is defined as
    \begin{align*}
        R_{\G, \mathrm{tot}} &:=\frac{1}{2}\sum_{a,b=1}^n R_{\G}(a, b) = \frac{1}{2}\mathbf{1}_n^{\top}R_{\G}\mathbf{1}_n.
    \end{align*}
     Moreover, if $\G$ is a strongly connected weight balanced directed graph, then $R_{\G}$ can be written as
    \begin{align*}
        R_{\G} = D_{\mathcal{L}_s^{\dagger}}\mathbf{1}_n \mathbf{1}_n^{\top} + \mathbf{1}_n \mathbf{1}_n^{\top} D_{\mathcal{L}_s^{\dagger}} - 2\mathcal{L}_s^{\dagger},
    \end{align*}
    where $D_{\mathcal{L}_s^{\dagger}} = \mathrm{diag}(\mathcal{L}_{s, 11}^{\dagger}, \ldots, \mathcal{L}_{s, n n}^{\dagger})$.
    Thus,
    \begin{align*}
        R_{\G, \mathrm{tot}} &= \frac{1}{2}\mathbf{1}_n^{\top}(D_{\mathcal{L}_s^{\dagger}}\mathbf{1}_n \mathbf{1}_n^{\top} + \mathbf{1}_n \mathbf{1}_n^{\top} D_{\mathcal{L}_s^{\dagger}} - 2\mathcal{L}_s^{\dagger})\mathbf{1}_n\\
        &= n \mathrm{tr}(\mathcal{L}_s^{\dagger}) = n \mathrm{tr}(\mathcal{L}^{\dagger}) = n \sum_{i=2}^n \lambda_i(\mathcal{L}^{\dagger}) = n\sum_{i=2}^{n}\frac{1}{\lambda_i(\mathcal{L})},
    \end{align*}
    where $\lambda_i(\mathcal{L})$ and $\lambda_i(\mathcal{L}^{\dagger})$ denote non-zero eigenvalues of $\mathcal{L}$ and $\mathcal{L}^{\dagger}$, respectively.
    In an undirected case, in terms of resistance distance, the total effective resistance is a quantitative measure of how well “connected" the network is, or how “large" the network is \cite{ghosh2008minimizing}.
\end{remark}

\begin{remark}
The pseudoinverse of $\mathcal{L}^{\dagger}_s$ is not equal to $\mathcal{L}_s := (\mathcal{L}+\mathcal{L}^{\top})/2$.
    Besides, the matrix $\left(\mathcal{L}^{\dagger}_s\right)^{\dagger}$ is a symmetric, undirected Laplacian on the same set of nodes and possibly admits negative edge weights.
\textcolor{black}{In fact, for example, consider
        $\mathcal{L} = \begin{bmatrix}
            2&-1&-1&0\\
            0&1&0&-1\\
            0&0&1&-1\\
            -2&0&0&2
        \end{bmatrix}$.
Then, $\left(\mathcal{L}^{\dagger}_s\right)^{\dagger} =
        \frac{1}{6}\begin{bmatrix}
            16&-4&-4&-8\\
            -4&7&1&-4\\
            -4&1&7&-4\\
            -8&-4&-4&16
        \end{bmatrix}
        \neq \mathcal{L}_s$.}
\end{remark}

Consider a Markov chain $X$, and define $C(a, b)$ as the commute time of $a$ and $b$; that is,
$C(a,b) = h(a, b) + h(b, a)$,
where
$h(a, b) := \mathrm{E}[\sigma_b \mid X_0 = a]$ is the mean hitting time from $a$ to $b$, which is
the expected number of steps from $a$ to $b$.
In Theorem 2.2 of \cite{chandra1996electrical}, the authors pointed out the relationship between effective resistance and commute time for undirected  connected graphs. 
This result can be generalized to strongly connected weight balanced directed graphs in the following way:
\begin{theorem}\label{thm: commute time}
    Let $\G = ([n], \E, \A)$ be a strongly connected weight balanced directed graph.
    Consider a Markov chain $X$ with probability transition matrix $P = \mathcal{D}^{-1}\A$ on $[n]$.
     Then,
    \begin{align*}
        C(a, b) &= (\mathrm{tr} \mathcal{D}) R_{\G}(a, b).
    \end{align*}
\end{theorem}
\begin{proof}
    For any vertex $i \in [n]$ such that $i \neq b$,
    \begin{align*}
        h(i, b) 
        &=1 + \sum_{j: \{i,j\}\textcolor{black}{\in \E}}P_{i j}h(j, b) = 1 + \frac{1}{\mathcal{D}_{i i}}\sum_{j: \{i,j\} \in \E}\A_{i j}h(j, b),
    \end{align*}
    where we used the Markov property and $P_{i j} = \A_{i j} / \mathcal{D}_{ii}$.
    Then,
    \begin{align*}
        \mathcal{D}_{ii} &= \mathcal{D}_{i i}h(i, b) - \sum_{j: \{i,j\} \in \E}\A_{i j}h(j, b) = \sum_{j: \{i,j\} \in \E}\A_{i j}(h(i,b) - h(j, b)).
    \end{align*}
    Let $h_b := (h(1, b), \ldots, h(n, b))^{\top}$. Then, for any vertex $i \in [n]$ satisfying $i \neq b$,
    \begin{align*}
        (\mathcal{L}h_b)_i &= \mathcal{D}_{ii}h_b(i) - \sum_{j: (i, j) \in \E}\A_{i j}h_b(j)
        = \sum_{j: \{i,j\}}\A_{i j}(h(i,b) - h(j, b))
        = \mathcal{D}_{i i}.
    \end{align*}
    Since $\G$ is weight balanced, $\mathbf{1}_n^{\top}\mathcal{L} = \mathbf{0}_n^{\top}$. Thus,
    \begin{align*}
        \sum_{i=1}^n (\mathcal{L}h_b)_i &= (\mathbf{1}_n)^{\top}\mathcal{L}h_b = \mathbf{0}_n^{\top}h_b = 0,
    \end{align*}
    which indicates that
    $(\mathcal{L}h_b)_b = -\sum_{i \neq b}\mathcal{D}_{i i} = \mathcal{D}_{b b} - \mathrm{tr} \mathcal{D}$.

    Similarly, if $h_a = (h(1, a), \ldots, h(n, a))^{\top}$, we then obtain $(\mathcal{L}h_a)_i = \mathcal{D}_{i i}$ when $i \neq a$ and
    $(\mathcal{L}h_a) = \mathcal{D}_{a a} - \mathrm{tr} \mathcal{D}$.
    Thus,
    \begin{align*}
        \mathcal{L}(h_b - h_a) &= \mathrm{tr} \mathcal{D}(e_a - e_b),
    \end{align*}
    According to Lemma 4 in \cite{fontan2021properties}, $\mathcal{L}^{\dagger}\mathcal{L} = \Pi_n$.
    Thus,
    \begin{align*}
        \frac{1}{\mathrm{tr} \mathcal{D}}\Pi_n(h_b -h_a) &= \mathcal{L}^{\dagger}(e_a - e_b).
    \end{align*}
    Therefore, we obtain
    \begin{align*}
        R_{\G}(a, b) &= (e_a - e_b)^{\top}\mathcal{L}^{\dagger}(e_a - e_b) \\
        &= \frac{1}{\mathrm{tr} \mathcal{D}}(e_a - e_b)^{\top}\Pi_n(h_b - h_a) = \frac{1}{\mathrm{tr} \mathcal{D}}(e_a - e_b)^{\top}(h_b - h_a)\\
        &=\frac{1}{\mathrm{tr} \mathcal{D}}(h(a,b)-h(a, a)-h(b,b)+h(b,a)) \\
        &= \frac{1}{\mathrm{tr} \mathcal{D}}(h(a,b)+h(b,a)) = \frac{1}{\mathrm{tr} \mathcal{D}}C(a,b),
    \end{align*}
    where we used
    Theorem \ref{thm: ER of Weight Balanced Directed Graph},
    $\Pi_n(e_a - e_b) = (e_a - e_b)$, and $h(a,a) = h(b,b) = 0$.
\end{proof}

We define $C$ as the cover time of $\G$; that is, $C = \max_{i \in V}C_i$,
where $C_i$ is the expected length of a walk starting at $i$ and ending when every vertex in $\G$ is visited at least once. The following corollary, which is a generalization of Theorem 2.4 in \cite{chandra1996electrical}, shows that the effective resistance captures the cover time.
\begin{corollary}
    Let $\G = ([n], \E, \A)$ be a strongly connected weight balanced directed graph. Consider a Markov chain $X$ with probability transition matrix $P = \mathcal{D}^{-1}\A$ on $[n]$. Then,
    \begin{align}
        \frac{\mathrm{tr}\mathcal{D}}{2}\max_{a,b \in [n]}R_{\G}(a,b) \leq C \leq (1 + \log n)\mathrm{tr}\mathcal{D}\max_{a,b \in [n]}R_{\G}(a,b) .
    \end{align}
\end{corollary}

\subsubsection{General strongly connected directed graphs}
\label{subsubsec: ER of general SC graphs}
In Theorem \ref{thm: ER of Weight Balanced Directed Graph}, we have
characterized the effective resistance of strongly connected weight balanced directed graphs using the corresponding loop-less Laplacian matrix,
and stated the symmetry in Theorem \ref{thm: the metricity}.
However,
the effective resistance for general strongly connected directed graphs is usually asymmetric; that is, $R_{\G}(i, j) \neq R_{\G}(j, i)$.
The following theorem shows that the asymmetry can be reduced to the node characteristics in a weight balanced directed graph.
\begin{theorem}[\textit{Effective Resistance of Strongly Connected Graph}]\label{thm: Effective Resistance of Strongly Connected Graph}
        Let $\G = ([n], \E, \A)$ be a strongly connected directed graph and let $\mathcal{L}$ be the corresponding loop-less Laplacian matrix.
        Then, there exists a positive diagonal matrix $M = \mathrm{diag}(m_1, \ldots, m_n)$ such that the effective resistance from node $a$ to node $b$ in $\G$ can be written as
        \begin{align*}
            R_{\G}(a, b)&= m_a(e_a - e_b)^{\top}(M\mathcal{L})^{\dagger}(e_a - e_b).
        \end{align*}
    \end{theorem}
\begin{proof}
    From Theorem 2 of \cite{altafini2019investigating}, there exists a unique (up to a scalar multiplication) positive diagonal matrix $M = \mathrm{diag}(m_1, \ldots, m_n)$
    such that $\ker(M\mathcal{L}) = \ker(\left(M\mathcal{L}\right)^{\top}) = \mathrm{span}(\mathbf{1}_n)$.
    Hence, $M\mathcal{L}$ is the loop-less Laplacian matrix of a weight balanced directed graph $\G_{\mathrm{balanced}} = ([n], \E, \A_{\mathrm{balanced}})$ such that
    $\A_{\mathrm{balanced}, i j} = m_i \A_{i j}$. Note that the Markov chains induced by $\G_{\mathrm{balanced}}$ and $\G$ are identical. For this reason,
    \begin{align*}
        R_{\G}(a, b) &= \frac{1}{\sum_{j=1}^{n}\A_{a j} P_{\mathrm{esc}}(a, b)}\\
        &= \frac{m_a}{\sum_{j=1}^n\A_{\mathrm{balanced}, a j} P_{\mathrm{esc}}(a, b)}\\
        &= m_a R_{\G_\mathrm{balanced}}(a, b)\\
        &= m_a(e_a - e_b)^{\top}(M\mathcal{L})^{\dagger}(e_a - e_b),
    \end{align*}
    where we used Theorem \ref{lem: probabilistic interpretation}.
    The last equality follows from Theorem \ref{thm: ER of Weight Balanced Directed Graph}.
\end{proof}
\begin{remark}
    Theorems \ref{thm: the metricity} and \ref{thm: Effective Resistance of Strongly Connected Graph} show that 
    \begin{align*}
        R_{\G_{\mathrm{balanced}}}(i, j) &= \frac{R_{\G}(i, j)}{m_i}
    \end{align*}
    is a metric on a strongly connected directed graph $\G$, which is invariant under a Kron reduction if the original graph is loop-less.
    Moreover, Theorem \ref{thm: the metricity} states that the metric $d_{i j} =\sqrt{R_{\G_{\mathrm{balanced}}}(i, j)}$ is a Euclidean distance.
    This is non-trivial, because we can show an example of a non-Euclidean distance on the graph.
    In fact, consider a simple 4-node undirected graph $G$ shown in Fig. \ref{fig: shortest_path_distance} and the shortest path distance on $G$.
    The matrix containing the shortest path distances taken pairwise between the elements of $G$ is obtained as follows:
    \begin{align*}
        D = \begin{bmatrix}
            0&1&\sqrt{2}&1\\
            1&0&1&2\\
            \sqrt{2}&1&0&1\\
            1&2&1&0
        \end{bmatrix}.
    \end{align*}
    Suppose that this structure can be embedded in a Euclidean space.
    Here, $D_{2 1} + D_{1 4} = D_{2 4}$ indicates that the values 2, 1, and 4 must be arranged in a straight line at equal intervals in that order.
    Similarly, $D_{2 3} + D_{3 4} = D_{2 4}$ implies that 2, 3, and 4 must be arranged in a straight line at equal intervals, also in that order.
    Thus, nodes 1 and 3 must be co-located, which contradicts $D_{1 3} = \sqrt{2} > 0$.
    Therefore, the shortest path distance on a graph is not always a Euclidean distance.
    \begin{figure}[hbtp]
    \centering
    \includegraphics[keepaspectratio, scale=0.27]{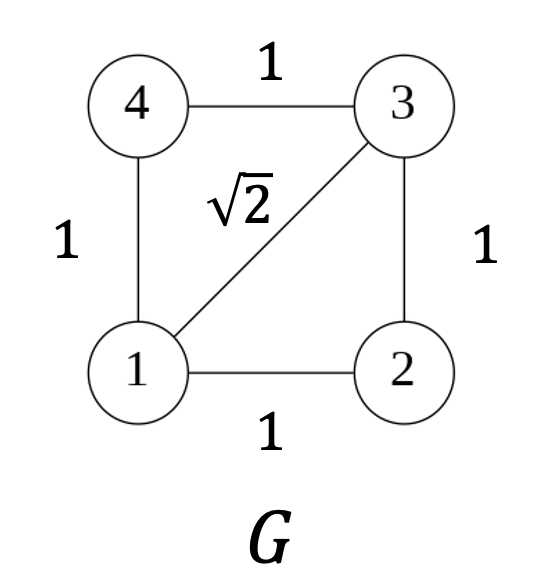}
    \caption{A simple 4-node undirected graph $G$.}
    \label{fig: shortest_path_distance}
\end{figure}
\end{remark}
\section{Comparison with existing studies}
\label{sec: Comparison with existing studies}
In this section, we discuss some of the concepts related to our study.
\subsection{Effective resistance and Kron reduction based on Lyapunov theory}
\label{subsec: ER and KR based on Lyapunov theory}
In this subsection, we introduce the effective resistance and Kron reduction based on Lyapunov theory.
\subsubsection{Effective resistance}
\label{subsubsec: ER based on Lyapunov theory}
The effective resistance of a directed graph based on Lyapunov theory has previously been defined as follows:
\begin{definition}(see \cite{young2015new})\label{def: ER_Lyapunov}
    Let $\G = ([n], \E, \A)$ be a directed graph without a self-loop, and let $\mathcal{L}$ be the corresponding Laplacian matrix.
    Suppose $G$ contains a globally reachable node.
    Then, the effective resistance between node $a$ and $b$ in $\G$ is defined as
    \begin{align*}
        \tilde{R}_{\G}(a, b) &:= (e_a - e_b)^{\top}X(e_a - e_b) = X_{a a} + X_{b b} - X_{a b} - X_{b a},
    \end{align*}
    where
    \begin{align}
        X &= 2P^{\top}\Sigma P,\,
        \Bar{L}\Sigma + \Sigma \Bar{L}^{\top} = I_{n-1},\,\label{Eq: Lyapunov}
        \Bar{L} = P\mathcal{L}P^{\top},
    \end{align}
    and $P \in \R^{(n-1)\times n}$ is a matrix satisfying
        $P\mathbf{1}_n = \mathbf{0}_{n-1}$, $P P^{\top} = I_{n-1}$, $P^{\top}P = \Pi_n$.
\end{definition}

Equation (\ref{Eq: Lyapunov}) has a unique, symmetric, and positive definite solution $\Sigma$ when $\G$ is connected, and $\Sigma$ is computed as
    $\Sigma = \int_0^{\infty}\exp(-\Bar{L}t)\exp(-\Bar{L}^{\top}t) \mathrm{d}t$.
Thus,
\begin{align*}
    \tilde{R}_{\G}(a, b) 
    = 2(e_a - e_b)^{\top}P^{\top}XP(e_a - e_b)
    = 2\int_0^{\infty}\left\| \left(e_{a} - e_{b}\right)^{\top}\exp(-\mathcal{L}t)\right\|^2 \mathrm{d}t.
\end{align*}

The following properties of the effective resistance were proven in \cite{young2015new}:
\begin{enumerate}
    \item The definition of $\tilde{R}_{\G}(a,b)$ is well-defined. In particular, the value is independent of the choice of $P$.\\
    \item The definition of $\tilde{R}_{\G}(a,b)$ is a generalization to that of the undirected effective resistance used in \cite{dorfler2012kron}; i.e., $X = L^{\dagger}$ if $L$ is symmetric.
    \item $\sqrt{\tilde{R}_{\G}(a,b)}$ is a graph metric, whereas $\tilde{R}_{\G}(a,b)$ is not.
\end{enumerate}

Unlike the effective resistance in the sense of Definition \ref{def: ER_Lyapunov},
the effective resistance introduced in this study is a graph metric for strongly connected weight balanced directed graphs,
as indicated in Theorem \ref{thm: the metricity}.

\subsubsection{Kron reduction method using symmetrization}
In \cite{fitch2019effective}, the author showed that the pseudoinverse of $X$ in Definition \ref{def: ER_Lyapunov} is a symmetric, positive semidefinite matrix with zero row and column sums.
Therefore, $\hat{L}_u := X^{\dagger}$ can be interpreted as an undirected Laplacian matrix, where $\hat{L}_u$ potentially admits negative edge weights and is associated with an undirected signed graph $\hat{G}_u$ on the same set of nodes as $\G$.
Here,
$(X^{\dagger})^{\dagger} = X$ implies that
    $\tilde{R}_{\G}(a, b) = (e_a - e_b)^{\top}(\hat{L}_u)^{\dagger}(e_a - e_b)$
for $a, b \in [n]$. 
Thus, the effective resistance in $\G$ is equivalent to that in $\hat{G}_u$.
Furthermore, $\mathcal{L}$ can be decomposed into the following matrices:
\begin{align*}
    \mathcal{L} &= H(I_n + 2K)\hat{L}_u,\label{Eq: from undirected to directed}
\end{align*}
where
\begin{align}
    H &= \mathcal{L}(\Pi_n\mathcal{L})^{\dagger},\label{Eq: def of H}\\
    (\Pi_n\mathcal{L}) K + K(\Pi_n\mathcal{L})^{\top} &= \frac{1}{2}((\Pi_n\mathcal{L}) - (\Pi_n\mathcal{L})^{\top})\label{Eq: def of K}.
\end{align}
Using this decomposition, a Kron reduction of directed graphs based on Lyapunov theory is defined as follows:
\begin{definition}(see \cite{fitch2019effective})\label{def: KR based on Lyapunov}
        Let $\G = ([n], \E, \A)$ be a directed graph with a globally reachable node and without a self-loop, and let $\mathcal{L}$ be the corresponding Laplacian matrix.
        Consider a subset of nodes $\alpha \subset [n]$. Subsequently, the Kron-reduced Laplacian $\mathcal{L}^{\mathrm{k\mathchar`-r}} \in \R^{|\alpha| \times |\alpha|}$
        is defined as
        \begin{align*}
            \mathcal{L}^{\mathrm{k\mathchar`-r}} &= H^{\mathrm{k\mathchar`-r}}(I_{|\alpha|} + 2K^{\mathrm{k\mathchar`-r}})\hat{L}_u^{\mathrm{k\mathchar`-r}},
        \end{align*}
        where
        \begin{align*}
            \hat{L}_u^{\mathrm{k\mathchar`-r}} &= \hat{L}_u/\hat{L}_u[\alpha^c, \alpha^c] = \hat{L}_u[\alpha, \alpha] - \hat{L}_u[\alpha, \alpha^c]\hat{L}_u[\alpha^c, \alpha^c]^{-1}\hat{L}_u[\alpha^c, \alpha],\\
            H^{\mathrm{k\mathchar`-r}} &= H[\alpha, \alpha]\Pi_{|\alpha|},\,\, K^{\mathrm{k\mathchar`-r}} = K[\alpha, \alpha]\Pi_{|\alpha|},
        \end{align*}
        and $H$ and $K$ are defined by (\ref{Eq: def of H}) and (\ref{Eq: def of K}), respectively.
\end{definition}

Let $\hat{G}_u$, $\hat{G}^{\mathrm{k\mathchar`-r}}_{u}$, $\G^{{\mathrm{k\mathchar`-r}}}$ be the corresponding graphs of $\hat{L}_u$, $\hat{L}_u^{\mathrm{k\mathchar`-r}}$, $\mathcal{L}^{\mathrm{k\mathchar`-r}}$, respectively.
Based on the construction, the effective resistance can be obtained from $\G$, $G$, $G^{\mathrm{k\mathchar`-r}}_{u}$, or $\G^{{\mathrm{k\mathchar`-r}}}$ for any $a, b \in \alpha$, i.e.,
\begin{align*}
    \tilde{R}_{\G}(a, b) &= \tilde{R}_{\hat{G}_u}(a, b) = \tilde{R}_{\hat{G}^{\mathrm{k\mathchar`-r}}_{u}}(a, b) = \tilde{R}_{\G^{{\mathrm{k\mathchar`-r}}}}(a, b).
\end{align*}
In other words, the effective resistance in Definition \ref{def: ER_Lyapunov} is invariant under the Kron reduction in Definition \ref{def: KR based on Lyapunov}.

\begin{figure}[hbtp]
    \centering
    \includegraphics[keepaspectratio, scale=0.20]{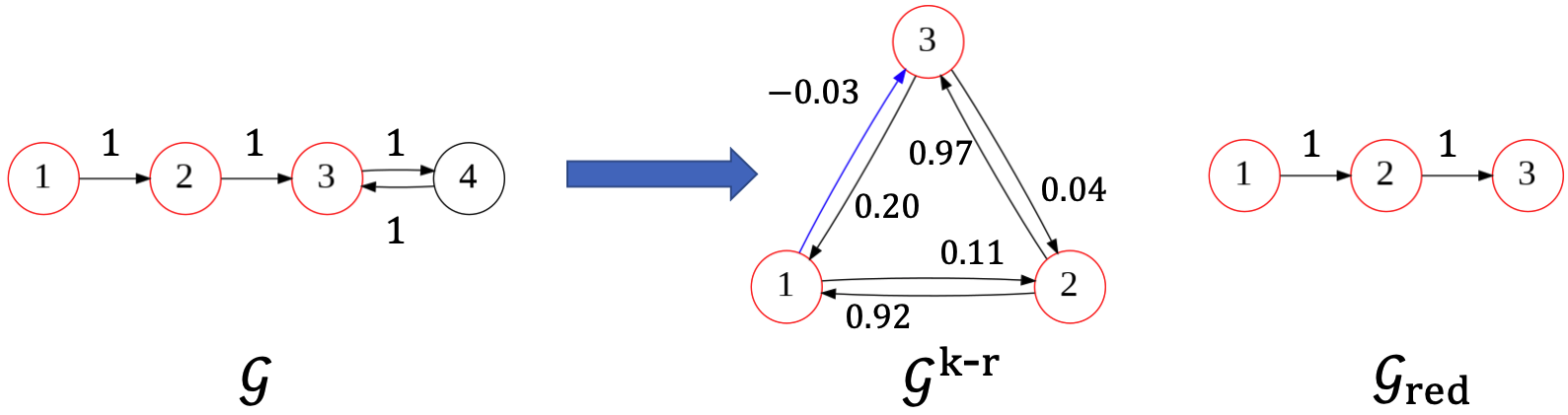}
    \caption{Illustration of Kron reduction of a 4-node directed graph $\G$. $\G^{\mathrm{k\mathchar`-r}}$ is calculated based on Definition \ref{def: KR based on Lyapunov}, and $\G_{\mathrm{red}}$ is calculated based on Definition \ref{def: KR}.}
    \label{fig: Comparison_with_KR}
\end{figure}

The Kron reduction method in Definition \ref{def: KR} preserves the original network structure, as mentioned in Section \ref{sec: Kron reduction of directed graphs}, whereas the approach based on Definition \ref{def: KR based on Lyapunov} does not.
To show an example of this, consider a 4-node graph and its Kron-reduced graphs illustrated in Fig. \ref{fig: Comparison_with_KR}.
The reduced graph $\G_{\mathrm{red}}$ defined in Definition \ref{def: KR} preserves the interconnection structure between the boundary nodes.
However, the reduced graph $\G^{\mathrm{k\mathchar`-r}}$ in Definition \ref{def: KR based on Lyapunov} is a complete directed graph having an edge with a negative weight.
Thus, the reduction method based on Definition \ref{def: KR based on Lyapunov} does not preserve the interconnection structure and non-negativity of the edge weights of the original graph.
Note that the Kron reduction shown in Definition \ref{def: KR} (Definition \ref{def: KR based on Lyapunov}) does not preserve the effective resistance in Definition \ref{def: ER_Lyapunov} (Definition \ref{def: EC and ER of directed graphs}), because our research adopts a different approach than those in \cite{fitch2019effective, young2015new}.

\subsection{Resistance distance of ergodic Markov chain}
\label{subsec: Resistance distance of ergodic Markov chain}
In this section, we introduce the \textit{resistance distance of an ergodic Markov chain} defined in \cite{choi2019resistance}.
\begin{definition}(see \cite{choi2019resistance})\label{def: Resistance distance of ergodic Markov chain}
        Let $(X_t)_{t=0}^{\infty}$ be an irreducible and aperiodic Markov chain in state space $[n]$ with transition matrix $P$,
        and let $\phi$ be an invariant distribution. By writing $\Phi$ as the matrix, where each row is $\phi$,
        the fundamental matrix $F \in \R^{n \times n}$ is given by
        \begin{align*}
            F := (I_n - P + \Phi)^{-1}.
        \end{align*}
        The resistance distance $\Omega_{i j}$ between $i, j \in [n]$ is defined as
        \begin{align*}
            \Omega_{i j} &:= F_{i i} + F_{j j} - F_{i j} - F_{j i}.
        \end{align*}
\end{definition}

Proposition 1.2 of \cite{choi2019resistance} states that $\Omega = (\Omega_{i j})_{i,j \in [n]}$ generally defines a semi-metric on $[n]$,
and is a metric when $P$ is doubly stochastic. 
By contrast, $\mathcal{L} := I - P$ can be interpreted as the loop-less Laplacian of graph $\G$ whose adjacency matrix is $P$.
Typically, the resistance distance is not equivalent to the effective resistance.
However, when $P$ is doubly stochastic and the corresponding graph is weight-balanced, the two concepts coincide.
To see this, note that Proposition 1.1 of \cite{choi2019resistance} states that $\Omega_{i j}$ can be written as
    $\Omega_{i j} = (e_i - e_j)D(e_i - e_j)$,
where $D$ is the group inverse of $\mathcal{L}$ satisfying
    $\mathcal{L}D\mathcal{L} = \mathcal{L}$, $D\mathcal{L}D = D$, $\mathcal{L}D = D\mathcal{L}$.
Since $\mathcal{L}$ is weight balanced, the pseudoinverse of $L$ satisfies
    $\mathcal{L}^{\dagger}\mathcal{L} = \mathcal{L} \mathcal{L}^{\dagger} = \Pi_n$.
The uniqueness of the group inverse \cite{erdelyi1967matrix} implies $\mathcal{L}^{\dagger} = D$. Thus,
\begin{align*}
    \Omega_{i j} &= (e_i - e_j)D(e_i - e_j) = (e_i - e_j)\mathcal{L}^{\dagger}(e_i - e_j) = R_{\G}(i, j).
\end{align*}
In other words, the resistance distance is the effective resistance defined in this study.
Thus, the resistance distance of an ergodic Markov chain is invariant under the Kron reduction in Definition \ref{def: KR} if the original graph is loop-less. Furthermore, Theorem \ref{thm: commute time} can be interpreted as a generalization of Proposition 1.1 (4) in \cite{choi2019resistance}.
\subsection{Hitting probability metric}
\label{subsec: Hitting probability metric}
In this subsection, we introduce the \textit{hitting probability metric} defined in \cite{boyd2021metric} and describe the relationship between the effective resistance and this metric.
\begin{definition}(see \cite{boyd2021metric})\label{def: Hitting probability metric}
        Let $(X_t)_{t=0}^{\infty}$ be a discrete-time Markov chain in state space $[n]$
        with an irreducible transition matrix $P$,
        and let $\phi$ be the invariant distribution for $P$; that is, $\phi P = \phi$.
        Then, the hitting probability metric $d^{\beta}: [n] \times [n] \to \R$ is defined as
        \begin{align*}
            d^{\beta}(i, j) &= -\log\left(A_{i, j}^{\mathrm{hp}, \beta}\right),
        \end{align*}
        where $A^{\mathrm{hp}, \beta} = \left(A_{i, j}^{\mathrm{hp}, \beta}\right)_{i, j \in [n]} \in \R^{n \times n}$ is a normalized hitting probability matrix defined as
        \begin{align*}
            A_{i, j}^{\mathrm{hp}, \beta} &:= \begin{cases}
                \frac{\phi_{i}^{\beta}}{\phi_j^{1-\beta}}P_{\mathrm{esc}}(i, j), &(i \neq j),\\
                1 &(i = j),
            \end{cases}
        \end{align*}
        where $\beta \in [1/2, \infty)$, and $P_{\mathrm{esc}}(i, j)$ is the probability, starting at $i$, that the walk reaches $j$ before returning to $i$.
    \end{definition}

Theorem 1.4 of \cite{boyd2021metric} states that $d^{\beta}$ is a metric for $\beta \in (1/2, 1]$.
Let $\G = ([n], \E, \A)$ be a strongly connected directed graph, $\mathcal{L}$ be the corresponding loop-less Laplacian matrix, and $P = \mathcal{D}^{-1}\A$
 correspond to the transition matrix. According to Theorem \ref{thm: Effective Resistance of Strongly Connected Graph}, there exists a unique (up to scalar multiplication) positive diagonal matrix $M$ such that
$\mathcal{L}_{\mathrm{balanced}} := M\mathcal{L}$ denotes the loop-less Laplacian matrix of a weight balanced directed graph $\G_{\mathrm{balanced}}$.
Let $M = \mathrm{diag}(m_1, \ldots, m_n)$ such that $\sum_{i=1}^n m_iD_{i i} = 1$. Then,
\begin{align*}
    (m_1D_{1 1}, \ldots, m_nD_{n n})P &= (m_1, \ldots, m_n)(\mathcal{D} - \mathcal{L})\\
    &= (m_1D_{1 1}, \ldots, m_nD_{n n}),
\end{align*}
where we used the equations $(m_1, \ldots, m_n)\mathcal{L} = \mathbf{1}_n^{\top}M\mathcal{L} = \mathbf{1}_n^{\top}\mathcal{L}_{\mathrm{balanced}} = \mathbf{0}_n^{\top}$.
The uniqueness of the invariant distribution indicates that $\phi_i = m_i D_{i i}$. Thus,
\begin{align*}
    d^{\beta}(i, j)
    &= -\log\left(\frac{\phi_i^{\beta}}{\phi_j^{1-\beta}}P_{\mathrm{esc}}(i, j)\right)\\
    &= -\log\left(\phi_i P_{\mathrm{esc}}(i, j)\right) + (1-\beta)\log\left(\phi_i\phi_j\right)\\
    &= -\log\left(m_i D_{i i}P_{\mathrm{esc}}(i, j)\right)+ (1-\beta)\log\left(\phi_i\phi_j\right)\\
    &= \log\left(R_{\G_{\mathrm{balanced}}}(i, j)\right)+ (1-\beta)\log\left(\phi_i\phi_j\right)\\
    &= \log\left(R_{\G}(i, j)\right) - \log \left(m_i\right)+ (1-\beta)\log\left(\phi_i\phi_j\right)
\end{align*}
for distinct $i, j \in [n]$, where we used Theorem \ref{lem: probabilistic interpretation} and Theorem \ref{thm: Effective Resistance of Strongly Connected Graph}.
These equations indicate that the effective resistance contains information regarding the hitting probability metric.
Thus, the hitting probability metric is also invariant under the Kron reduction in Definition \ref{def: KR} if the original graph is loop-less.
In particular, $d^1$ is the logarithm of the effective resistance of the corresponding weight balanced directed graph $\G_{\mathrm{balanced}}$.

\section{Conclusion}
\label{sec: Conclusion}
We have generalized Kron reduction method to directed graphs in a manner that maintains the original interconnection structure and weight-balance.
We have also proposed the effective resistance of directed graphs based on the Markov chain theory, which is invariant under a Kron reduction.
Our generalization is derived from a probabilistic interpretation in which directed graphs arise naturally and can be calculated based on the Moore-Penrose pseudoinverse of a loop-less Laplacian.
Furthermore, we have shown that when a graph is strongly connected and weight balanced, the effective resistance relates to commute and covering times.
In addition, the effective resistance induces two novel distances on general strongly connected directed graphs.
We have proved that the resistance distance on an ergodic Markov chain is the same with the effective resistance in this paper in the doubly stochastic case, and we
have shown an association between the hitting probability metric and the effective resistance under a stochastic case.
Because the effective resistance and induced metric are insensitive to the shortest path distances, we believe that they can provide new information for applications such as a graph visualization, data exploration, and cluster detection, which will be a topic for future research.
Another topic for future research is a generalization of the relationship among the effective resistance, commute and covering times to general strongly connected directed graphs.
\bibliographystyle{siamplain}
\bibliography{references}
\end{document}